\documentclass[sigconf]{aamas} 

\setcopyright{none}
\renewcommand\footnotetextcopyrightpermission[1]{} 

\usepackage{balance} 

\usepackage{soul}
\usepackage{url}
\usepackage{graphicx}
\usepackage{amsmath}
\usepackage{booktabs}
\usepackage{algorithm, algpseudocode}
\usepackage{url}
\usepackage{graphicx}
\usepackage{xcolor}
\usepackage{tikz}

\usepackage{mathtools}

\usepackage{soul}
\usepackage{todonotes}

\usepackage{ stmaryrd }

\usepackage{cleveref}

\usepackage{stackengine}

\DeclareMathAlphabet{\mathcal}{OMS}{cmsy}{m}{n}

\newcommand{\ourlogic}{\mathsf{SPQ}}

\newcommand{\Terms}{\mathsf{Terms}}
\newcommand{\Prop}{\mathsf{Prop}}
\newcommand{\Cost}{\mathsf{Cost}}
\newcommand{\Bdg}{\mathsf{Bdg}}
\newcommand{\cL}{\mathcal{L}}
\newcommand{\cM}{{\mathcal{M}}}

\newcommand{\LPL}{\cL_{PL}}
\newcommand{\ZZ}{\mathbb{Z}}

\newcommand{\AG}{\mathbb{AG}}

\newcommand{\Sub}{\mathsf{Sub}}

\algtext*{EndWhile}
\algtext*{EndIf}
\algtext*{EndFor}

\algnewcommand\algorithmiccase{\textbf{case}}
\algdef{SE}[CASE]{Case}{EndCase}[1]{\algorithmiccase\ #1}{\algorithmicend\ \algorithmiccase}%
\algtext*{EndCase}

\makeatletter
\newenvironment{breakablealgorithm}
  {
   \begin{center}
     \refstepcounter{algorithm}
     \hrule height.8pt depth0pt \kern2pt
     \renewcommand{\caption}[2][\relax]{
       {\raggedright\textbf{\fname@algorithm~\thealgorithm} ##2\par}%
       \ifx\relax##1\relax 
         \addcontentsline{loa}{algorithm}{\protect\numberline{\thealgorithm}##2}%
       \else 
         \addcontentsline{loa}{algorithm}{\protect\numberline{\thealgorithm}##1}%
       \fi
       \kern2pt\hrule\kern2pt
     }
  }{
     \kern2pt\hrule\relax
   \end{center}
  }
\makeatother

\usetikzlibrary{calc}
\usetikzlibrary{arrows}
\usetikzlibrary{positioning}

\usetikzlibrary{chains,fit,shapes}

\tikzset{%
round/.style={circle, draw=gray!60,fill=gray!5, very thick,minimum size=7mm, align=center},
dot/.style={draw, circle, minimum size=1mm,inner sep=0pt,outer sep=0pt,fill=black},
}

\acmSubmissionID{254}

\title{Dynamic Epistemic Logic of Resource Bounded Information Mining Agents}
\titlenote{This is a slightly extended version of the same title paper that will appear in AAMAS
2024. This version contains a small appendix with proofs that, for space reasons, do
not appear in the AAMAS 2024 version.
}

\author{Vitaliy Dolgorukov}
\affiliation{
  \institution{HSE University}
  \city{Moscow}
  \country{ Russian Federation}}
\email{vdolgorukov@hse.ru}

\author{Rustam Galimullin}
\affiliation{
  \institution{University of Bergen}
  \city{Bergen}
  \country{Norway}}
\email{rustam.galimullin@uib.no}

\author{Maksim Gladyshev}
\affiliation{
  \institution{Utrecht University}
  \city{Utrecht}
  \country{Netherlands}}
\email{m.gladyshev@uu.nl}

\begin{abstract}
Logics for resource-bounded agents have been getting more and more attention in recent years since they provide us with more realistic tools for modelling and reasoning about multi-agent systems. While many existing approaches are based on the idea of agents as imperfect reasoners, who must spend their resources to perform logical inference, this is not the only way to introduce resource constraints into logical settings. 
In this paper we study agents as perfect reasoners, who may purchase a new piece of information from a trustworthy source. For this purpose we propose dynamic epistemic logic for semi-public queries for resource-bounded agents. In this logic (groups of) agents can perform a query (ask a question) about  whether some formula is true and receive a correct answer. These queries are called semi-public, because the very fact of the query is public, while the answer is private. We also assume that every query has a cost and every agent has a budget constraint. Finally, our framework allows us to reason about group queries, in which agents may share resources to obtain a new piece of information together. We demonstrate that our logic is complete,
decidable and has an efficient model checking procedure.
\end{abstract}


\newcommand{\BibTeX}{\rm B\kern-.05em{\sc i\kern-.025em b}\kern-.08em\TeX}

\makeatletter
\gdef\@copyrightpermission{
	\begin{minipage}{0.3\columnwidth}
		\href{https://creativecommons.org/licenses/by/4.0/}{\includegraphics[width=0.90\textwidth]{figs/by.eps}}
	\end{minipage}\hfill
	\begin{minipage}{0.7\columnwidth}
		\href{https://creativecommons.org/licenses/by/4.0/}{This work is licensed under a Creative Commons Attribution International 4.0 License.}
	\end{minipage}
	\vspace{5pt}
}
\makeatother

\begin{document}


\pagestyle{fancy}
\fancyhead{}



\maketitle


\section{Introduction}

 In our paper we present a logic for reasoning about resource bounded agents who might purchase new information from a reliable source but do not necessarily have enough budget resources. 

From a technical perspective, we propose a multi-agent Kripke-style semantics, in which propositional formulas have (non-negative) costs and agents have (non-negative) budgets, i.e. amounts of available resources. In order to deal with costs and budgets in our language, we use linear inequalities initially introduced in \cite{HalpernInequalities} for reasoning about probabilities. These inequalities allow us to formulate statements about formulas' costs, agents' budgets, and their comparisons explicitly as formulas of our language. 
We also use standard epistemic primitives for individual and common knowledge \cite{HalpernBook}. Combined with linear inequalities, these primitives allow us to express statements of the form ``agent $i$ knows that her budget is at least $k$", ``agent $j$ knows that the cost of $A$ is lower than the cost of $B$, but $j$ does not know the cost of $A$", ``it is common knowledge among group $G$ that the joint budget of another group $D$ is lower than the cost of $A$", etc. Finally, we introduce a dynamic operator in the vein of dynamic epistemic logic (DEL) \cite{DELbook} for (group) semi-public queries. This operator assumes that a group of agents $G$ may perform the following action. Let $A$ be some propositional formula and $\varphi$ be a formula of our logic introduced below. We assume that there is a reliable source of information, such that $G$ may perform a query to this source, ask ``Is $A$ true?" and receive a correct answer: either 'Yes' or 'No'. Once this answer is received, $\varphi$ holds. Such queries are semi-public in the sense that the answer is private, i.e. available to members of $G$ only, but the very fact of the query is public, i.e. other agents from $\AG - G$ (where $\AG$ is the set of all agents) observe the fact that $G$ has performed a query, but do not observe the answer. 

There is also an additional constraint in our framework: formula $A$ has a cost, and this cost can be different for different agents in $G$. The first assumption captures a natural intuition that access to the information is not always free, while the latter shows that this access can be non-symmetric (or non-egalitarian) among agents. For example, there can be premium access to a database, so the same query may have a lower cost for agents with this access; one lab may have  access to cheaper reagents than another, so it may perform the same test spending less resources and so on. So, in this settings it may be rational for agents to cooperate and optimize the amount of resources they need to obtain a certain piece of knowledge. The last assumption we introduce in this paper is that agents may share resources in groups. Thus, if $G$ decides to perform a query $A$, they identify the agent $i\in G$ for whom the cost of this query is the lowest and then share an equal amount of resources to perform the query. Such group queries are expressed as an operator $[?_G^A]\varphi$ of our logic. This operator is inspired by the (dynamic) epistemic logics of contingency \cite{fan_wang_ditmarsch_2015,vanDitmarsch-Fan}.  These logics focus on the notion of ‘knowing whether’, which clearly describes an epistemic attitude of agents after the query $[?_G^A]$: all agents $j\in \AG\setminus G$ know that all $i\in G$ know whether $A$ is true. 
An earlier version of our logic restricted to individual agents was proposed in \cite{Dali2022}.
The idea of group updates in dynamic epistemic logic was proposed and studied in, for example, \cite{Agotnes_vanDitmarsch2008,AGOTNES2010,Galimullin2021,Galimullin_2023}.

This paper is organized as follows. In Section 2 we introduce the language and models of our logic $\ourlogic$. In Section 3 we propose a sound and complete axiomatisation for $\ourlogic$. In Section 4 we provide a polynomial time algorithm for solving the global model checking problem for $\ourlogic$ and demonstrate that the satisfiability problem for $\ourlogic$ is decidable. Finally, in Section 5 we overview existing works in this field, and in Section 6
we discuss open problems and possible directions for future work. 

\section{Logic of Semi-public Queries}

\subsection{Language}


At first, we need to fix a propositional language $\LPL$. 
Let $\Prop$  denote a countable set of
propositional letters $\{p, q, \dots\}$. Language $\LPL$ is defined
 by the following grammar: 
\[
A \ \ ::= \ p\ \mid\ \neg A\ \mid\ (A \land A).
\]
Here $p \in \Prop$, and all the usual abbreviations of propositional logic (such as $\top, \bot, \to$, and $\vee$) hold. Let $\AG = \{a_1, \dots, a_k\}$ be a finite set of agents.
We fix a set of terms
\[\Terms = \{c_{(A, i)}\mid A\in \LPL, i\in\AG\}\cup \{b_i\mid i\in \AG\}\]
It contains a special term $c_{(A, i)}$ for the \emph{cost} of each propositional formula~$A$ for agent $i$,
and a term $b_i$ for the \emph{budget} of each agent~$i$. So, we assume that the same formula $A$ can have a different cost for different agents. Now, we can define our language $\cL_{\ourlogic}$. 

\begin{definition}[Language]\label{def:syntax} The \emph{language $\cL_{\ourlogic}$ of Epistemic Logic for Semi-Public Queries} is defined recursively as follows
\[\varphi :: = p \mid (z_1t_1+\dots + z_nt_n) \geq z)\mid
\neg\varphi \mid (\varphi \wedge\varphi) \mid K_i \varphi \mid C_G\varphi\mid [?_G^A]\varphi,\]
where $p\in \Prop$, $t_1, \dots , t_n \in \Terms$, ${z_1, \dots , z_n, z\in \ZZ}$, $i\in\AG$, $G\subseteq\AG$, and $A\in\LPL$.

\end{definition}

Here $K_i\varphi$ means \textit{ ``agent $i$ knows that $\varphi$ is true"}, and $C_G\varphi$ means \textit{``it is common knowledge among agents in $G$ that $\varphi$ is true"}. These are two standard operators for epistemic logic \cite{HalpernBook}. Linear inequalities of the form $(z_1t_1+ \dots + z_nt_n)\geq z$ \cite{HalpernInequalities}, in which only the terms $c_{(A, i)}$ and $b_i$ from $\Terms$ can occur, allow us to reason about formulas' costs and agents' budgets explicitly. 
Finally, we interpret the dynamic operator  $[?_G^A]\varphi$ as \textit{"after a group query by $G$ whether formula $A$ is true, $\varphi$ is true"}. This operator can alternatively be understood as ``\textit{if $G$ performs query $A$, they can achieve $\varphi$}".

The dual 
of $K_i$ is 
$\hat{K}_i \varphi := \neg K_i \neg \varphi$.
Abbreviation $E_G\varphi:= \bigwedge_{i\in G} K_i\varphi$ means \textit{``everybody in $G$ knows $\varphi$"}.  The dual for dynamic operator is $\langle?_G^A\rangle \varphi := \neg [?_G^A]\neg\varphi$.  For linear inequalities we use  the same abbreviations as in~\cite{HalpernInequalities}.
Thus, we write $t_1 - t_2 \ge z$ for $t_1 + (-1)t_2 \ge z$,
$t_1 \ge t_2$ for $t_1 - t_2 \ge 0$,
$t_1\leq z$ for $-t_1\ge -z$, $t_1<z$ for $\neg (t_1 \ge z)$,
and $t_1 = z$ for $(t_1\ge z)\land (t_1\leq z)$. A formula of the form $t\geq \frac{1}{2}$ can be viewed as an
abbreviation for $2t\ge 1$, so we allow rational numbers to appear in our
formulas. 
Other Boolean connectives $\to, \lor, \leftrightarrow, \bot$ and $\top$ are defined in the standard way.  In the rest of the paper we slightly abuse the notation and write $c_i(A)$ instead of $c_{(A, i)}$. We denote the set of \emph{subformulas} of a formula $\varphi$ as $\Sub(\varphi)$ .

Thus, the static fragment of our language $\cL_{\ourlogic}$ (i.e. the fragment without $[?_G^A]\varphi$) allows us to express statements of the form 
$c_i(p \vee q) \ge 10$ for \textit{``the cost of the query whether $p$ or $q$ is true for agent $i$ is at least 10"}, $b_j \ge 3$ for \textit{``the budget of agent $j$ is at least 3"}, $2 b_j = b_i$ for \textit{"$i$'s budget is twice as big as that of $j$"}, $K_a (b_i + b_j) \ge c_i(p \vee q)$ for \textit{``agent $a$ knows that the joint budget of $i$ and $j$ is higher than the cost of $p\vee q$ for agent $i$"}, etc. The dynamic operator can express statements like $[?_{\{i, j\}}^{(p\vee q)}]C_{\{i, j\}}\neg p$ meaning that \textit{``after a joint query about $p\vee q$ by $\{i, j\}$ it is common knowledge among $\{i, j\}$ that $p$ is false"}. 

\subsection{Semantics}
 
Now we are ready to discuss the semantics for $\cL_{\ourlogic}$ formulas. Models of our logic are basically Kripke-style models endowed with Cost and Budget functions.

\begin{definition}[Model]\label{def:model} A \emph{model} is a tuple $\cM = (W, (\sim_i)_{i\in \AG},$ $\Cost,$ $\Bdg,$ $V)$, where
\begin{itemize}
    \item $W$ is a non-empty set of \emph{states},
    \item ${\sim_i} \subseteq (W\times W)$ is an equivalence relation for each ${i\in \AG}$,
    \item $\Cost\colon \AG \times W\times \LPL \longrightarrow \mathbb{Q}^+\cup \{0\}$ assigns the (non-negative) \emph{cost} to propositional formulas for each agent in each state,
    \item $\Bdg\colon \AG \times W \longrightarrow \mathbb{Q}^+\cup \{0\}$ is the (non-negative) \emph{bugdet} of each agent at each state,
    \item $V\colon \Prop \longrightarrow 2^W$ is a \emph{valuation} of propositional variables.
\end{itemize}
Let $P (w) := \{p \in \Prop \mid w \in V(p)\}$ be the set of all propositional variables that are true in state $w$. Moreover, let $c_{>0} (i,w) := \{A \in \mathcal{L}_{PL} \mid \Cost_i(w,A)>0\}$ be the set of all propositional formulas with positive cost for agent $i$ and state $w$. We call a model \emph{finite}, if all of $W$, $\bigcup_{w \in W} P(w)$, and $\bigcup_{(i,w) \in \AG \times W} c_{>0} (i,w)$ are finite. 
Given a finite model $\cM$ we define the \emph{size} of $\cM$, denoted $|\cM|$, as 
\begin{multline*}
    |\cM|:= \mathrm{card}(W) + \sum_{i \in \AG}\mathrm{card}(\sim_i) +\\ + \sum_{(i,w) \in \AG \times W} \mathrm{card} (c_{>0}(i,w)) + \sum_{w \in W}\mathrm{card}(P(w))
\end{multline*}
\end{definition}

 We intentionally put as few restrictions on the $\Cost$ function as possible to consider the most general case. Thus, our framework allows us to model situations in which the cost of the same formula is different for different agents and it can also be different across different states in $W$ for the same agent. Thus, the agent may be unaware of the cost of some formula for herself as well as for other agents. Another important point is how the costs of different formulas must be related to each other \cite{Feldman2000}. 
The only two restrictions we find important to enforce are that the cost of propositional
tautologies must be zero for all agents, and that the costs of \emph{similar} formulas must be the same. By \emph{similar} formulas $A$ and $B$ (denoted $A\approx B$) we mean that \ $A \approx B$ \ iff \ $A\equiv B$ or $A\equiv \neg B$, where $A\equiv B$ denotes \emph{equivalent} formulas: $A\equiv B$ iff $\vdash_{PL}A\leftrightarrow B$.  Formally, for all $i\in \AG, w\in W$ we require that\\
(C1) $\Cost_i(w, \top) = 0$,\\
(C2) $A\approx B$ implies $\Cost_i(w, A)=\Cost_i(w, B)$.

For the $\Bdg$ function we only assume that it is non-negative for every agent. But the budget of each agent may be different in different states of $W$, so in our framework agents may be unaware of their and others' budgets. To better illustrate the proposed semantics consider a simple example.

\begin{example}[Telescope example]\label{simpleexample} Three countries $n, m$ and $l$ are seeking to know a certain fact $p$ about the universe. If any of them build a very expensive telescope, it will give them a correct answer. Country $n$ is the richest among others having 15 abstract resources ($b_n=15$). But due to some reasons, e.g. higher labour costs, it requires the highest amount of resources $c_n(p)=30$ to build a telescope there. Country $m$ has only 10 resources ($b_m=10$), but it can build a telescope for $c_m(p)=20$, for example due to better logistics. Country $l$ is the poorest country with only $b_l=5$ resources, while the cost of a telescope is the same as for $n$, so $c_l(p)=30$. Finally, we assume that the costs of the telescope are known to all agents, $n$ and $m$ know the budgets of each other as well as $l$'s budget. But $l$ is unaware of the exact budget of $m$: it considers both options $b_m=10$ and $b_m = 9$ as possible ones. 
\end{example}

\begin{figure}[t]
\resizebox{.47\textwidth}{!}{%
    \begin{tikzpicture} 
    \node[dot] (w1) at (-3, 0) [label=below left:\underline{$w_1$}] {};
    \node (w1') [above=.1cm of w1] {$p, b_m=10$};

    \node[dot] (w2) [label=below left:$w_3$] [right=2.5cm of w1] {};
    \node (w2') [above=.1cm of w2] {$p, b_m=9$};

    \node[dot] (w3) [label=left:$w_2$] [below=2cm of w1] {};
    \node (w3') [below=.1cm of w3] {$\neg p, b_m=10$};

    \node[dot] (w4) [label=above left:$w_4$] [below=2cm of w2] {};
    \node (w4') [below=.1cm of w4] {$\neg p, b_m=9$};

     \draw[latex-latex, dashed] (w1) edge node[anchor=south,midway,sloped, above] {$l, m, n$} (w3);
     \draw[latex-latex, dashed] (w2)  edge node[anchor=south,midway,sloped, above] {$l, m, n$}  (w4);

      \draw[latex-latex, dashed] (w1) edge node[anchor=south,midway,sloped, below] {$l$} (w2);

       \draw[latex-latex, dashed] (w3) edge node[anchor=south,midway,sloped, below] {$l$} (w4);

    \node at (1, -1) {$\xRightarrow{?_{\{m,n\}}^p}$};


    \node[dot] (w1) at (2.3, 0) [label=left:$\underline{w_1}$] {};
    \node (w1') [above=.1cm of w1] {$p, b_m=0$};

    \node[dot] (w2) [label=left:$w_2$] [below=2.1cm of w1] {};
    \node (w2') [below=.1cm of w2] {$\neg p, b_m=0$};

     \draw[latex-latex, dashed] (w1) edge node[anchor=south,midway,sloped, above] {$l$} (w2);

    \node at (-1.8, -3.1) {$\mathcal{M}$};
    \node at (2.3, -3.1) {$\mathcal{M}^{?_{\{m,n\}}^p}$};
    \end{tikzpicture}
    }
    \caption{(Left) Model $\cM$ for \Cref{simpleexample}. The model contains four states $w_1, w_2, w_3$ and $w_4$, and we assume that $w_1$ is the actual one. Arrows represent epistemic equivalence classes for agents $l, m$ and $n$. Reflexive and transitive arrows are omitted for readability. Formulas $b_n=15, b_l=5, c_n(p)=30, c_m(p)=20$ and $c_l(p)=30$ hold in all four states, and we omit them from the figure.  (Right) Updated model $\cM^{?_{\{n, m\}}p}$ for \Cref{simpleexample}. Formulas $b_l=5, c_n(p)=30, c_m(p)=20$ and $c_l(p)=30$ hold in both $w_1$ and $w_2$ as before. But $b_n=15$ no longer holds, since $n$'s budget is decreased after update: $b_n-\frac{c_m(p)}{|\{n, m\}|}=5$. 
    } 
    \label{fig:simpleexample}
\end{figure}
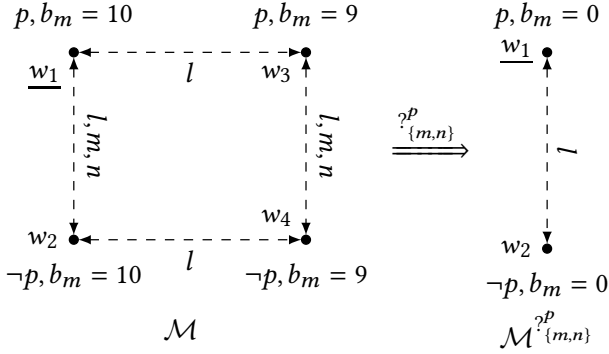

The model of this example is depicted in \Cref{fig:simpleexample}.
In this example, no single country has a sufficient budget to build the telescope, i.e. to perform an individual query about $p$. But there is still a way for them to cooperate and build a telescope. If $n$ and $m$ share their resources, and $m$ builds the telescope, then each of the countries $m$ and $n$ can spend 10 resources to get the information about $p$. This procedure can be expressed by our dynamic operator $[?_{\{n, m\}}^p]$. 

Now we are ready to discuss the semantics of semi-public group queries $[?_G^A]$. 
Recall that $[?_G^A]\varphi$ means \textit{``after a group query by $G$ whether a formula $A$ is true, $\varphi$ is true”}. We assume that each group member receives the correct information about the truth of $A$ after $[?_G^A]$. And we also assume that each group member spends equal amount of resources on this query. But as we already mentioned, the cost of $A$ can be different for different members of $G$.  So, it is natural to assume that the lowest of these costs must be spent. In other words, $[?_G^A]$ query in a state $w$ can be described by the following procedure: identify $i\in G$ with the lowest cost of $A$, let each member of $G$ transfer $\frac{Cost_i(w, A)}{|G|}$ resources to $i$, then let $i$ ask whether $A$ is true and tell the answer to all agents in $G$.

More formally, let us abbreviate \textit{``the Budget Constraint of agent $i\in G$ for the $G$'s query $A$"} as \[\mathbf{BC}_i(G, A) \equiv \frac{\min_{j\in G}(c_j(A))}{|G|}\] So, $\mathbf{BC}_i(G, A)$ denotes the budget that would be sufficient for $i$ to participate in a $G$'s group query whether $A$ is true. 

We also denote the fact that \textit{``the Budget Constraint for the query $A$ for $G$ is Satisfied"} as
\[\text{BCS}(G, A)\equiv \bigwedge\limits_{i\in G}(b_i \geq \mathbf{BC}_i(G, A))\] 
If $\text{BCS}(G, A)$ holds, we say that the query $[?_G^A]$ is realisable meaning that each group member has enough resources to cooperate according to our resource distribution rule\footnote{Note that linear inequalities in our language are capable enough to express alternative resource distribution rules. For example, we can assume that according to another rule $\text{BCS}'(G, A)$ we pick the highest, but not the lowest cost among members of $G$. This rule looks more arguable, but it is clearly expressible in $\ourlogic$. So, any expressible $\text{BCS}'(G, A)$ can be integrated in our framework with minimal changes.}. Note that $\text{BCS}(G, A)$ is in fact a formula of $\ourlogic$:
\[\text{BCS}(G, A)\equiv \bigvee\limits_{j\in G}\biggl(\bigwedge\limits_{i\in G}\Bigl(c_j(A)\leq c_i(A) \wedge b_i\geq \frac{c_j(A)}{|G|}\Bigl)\biggl)\]


\begin{definition}[Updated Model]\label{def:updmodel} Given a model $\cM$, a group $G\subseteq \AG$ and a formula $A\in \LPL$, an \emph{updated model} $\cM'$ is a tuple
$\cM'= (W',$ $(\sim'_j)_{j\in\AG},$ $\Cost',$ $\Bdg',$ $V')$, where
\begin{itemize}
\item $W' = \{ w \in  W \mid  \cM, w \vDash \text{BCS}(G, A) \}$;
\item $
{\sim_j'} = (W' \times W') \cap {\sim^*_j}$ with
\[\sim^*_j =  \begin{cases} \sim_j & \mbox{if } j\notin G,\\
{\sim_j}\bigcap 
\Bigl(([A]_\cM \times [A]_\cM) \bigcup ([\neg A]_\cM \times [\neg A]_\cM)\Bigr) & \mbox{if } j\in G;\end{cases}\]
\item $\Cost_j'(w, B) = \Cost_j(w, B)$, for all $B\in \LPL, j\in \AG$;
\item $\Bdg_j'(w) = \begin{cases} \Bdg_j(w) - \frac{\min\limits_{i\in G}\Cost_i(w, A)}{|G|},&\mbox{if } j\in G, \\
\Bdg_j(w), &\mbox{if } j\notin G,\end{cases}$
\item $V'(p) = V(p) \cap W'$ for all $p\in \Prop$.
\end{itemize}
By $[A]_{\cM}$ we denote the set of states in $W$, such that each $w\in [A]_{\cM}$ satisfies the formula $A$ in a sense of \Cref{def:semantics}.
\end{definition}

Intuitively, an update $[?_G^A]$ of a model $\cM$ firstly removes all states of $\cM$ in which at least one agent in $G$ does not have a sufficient amount of resources for a $G$'s query about $A$. This can be justified by the fact that agents do not necessarily know others budgets, but when they observe the fact that $G$ actually performs a query $A$, it no longer makes sense to consider the states in which $\text{BCS}(G, A)$ does not hold as possible ones. Note also that when $G$ performs the semi-public query \textit{``is $A$ true?"}, it gets either \textit{'Yes'} or \textit{'No'} as an answer and we consider this fact to be known by all agents. Then, after the update, all agents in $G$ necessarily distinguish any two states of $\cM$ that do not agree on the valuation of $A$. But since the actual answer is available only to the agents from $G$, epistemic relations of other agents remain the same, only taking into account that some states have been removed. This update does not affect the costs of formulas and budgets of all agents outside of $G$. The budget of each $i\in G$ decreases by the minimal cost of $A$ in the group divided by the size of this group. 

Returning to \Cref{simpleexample}, consider the result of updating $\cM$ with $[?_{\{n, m\}}p]$ in \Cref{fig:simpleexample}. 
After this update, both $n$ and $m$ will know that $p$ is true, moreover it will be common knowledge among them. Note that since $l$ is not a member of this group, she will remain unaware of whether $p$ is true or not. But since the fact that the telescope is built is public, $l$ will know that $n$ and $m$ know whether $p$ is true, denoted $K_l(E_{\{n,m\}}p\vee E_{\{n,m\}}\neg p)$. Moreover, $l$ will know that $\{n, m\}$ has common knowledge whether $p$ is true, denoted $K_l(C_{\{n,m\}}p\vee C_{\{n,m\}}\neg p)$.
This is why we call these queries semi-public. Note also that budgets of $n$ and $m$ are also decreased by 10 according to the resource distribution rule. 

We call two states $x, y\in W$ $G$-\emph{reachable} iff there is a sequence of states $w_0, \dots, w_k$ (where $k\geq 1$), such that $w_0=x$ and $w_k=y$ and for all $0\leq j \leq k-1$ there exists $i\in G$ such that $w_j\sim_iw_{j+1}$. We denote the  $G$-reachability relation as $\sim_G$. 

\begin{definition}[Semantics]\label{def:semantics}  The \emph{truth} $\vDash$ of a formula $\psi\in\cL(\ourlogic)$
at a state $w\in W$ of a model $\cM$ is defined by induction:\\
$\cM, w\vDash p$ iff $w\in V(p)$, \\
$\cM, w\vDash \neg \varphi$ iff $\cM, w\nvDash \varphi$,\\
$\cM, w\vDash \varphi \land \psi$ iff $\cM, w\vDash \varphi$ and $\cM, w \vDash \psi$, \\
$\cM, w\vDash K_i\varphi$ \ iff \ $\forall w'\in W$: $w\sim_i w' \Rightarrow \cM, w'\vDash \varphi$,\\
$\cM, w\vDash C_G\varphi$ \ iff \ $\forall w'\in W$: $w\sim_G w' \Rightarrow \cM, w'\vDash \varphi$,\\
$\cM, w\vDash (z_1t_1+\dots + z_nt_n) \ge z$ iff $(z_1t'_1+\dots + z_nt'_n) \ge z$,
where for $1\leq k \leq n$,
\begin{center}
$t'_k =
  \begin{cases}
    \Cost_i(w, A), & \text{for } t_k = c_i{(A)} \\
    \Bdg_i(w), & \text{for } t_k = b_i. 
  \end{cases}$
\end{center}
$\cM, w\vDash [?_G^A]\varphi$ \quad iff \quad $\cM, w\vDash \text{BCS}(G, A)$ implies $\cM^{?_G^A}, w\vDash \varphi$, where $\cM^{?_G^A}$ is an updated model in the sense of \Cref{def:updmodel}.
\end{definition}



    



$\ourlogic$ provides us with tools for reasoning about 
coalition formation strategies to obtain the desirable knowledge by group queries. As \Cref{simpleexample} demonstrates, sometimes no single agent has enough resources to achieve certain knowledge, but this knowledge is still achievable via cooperation. Moreover, in many situations not only the choice of coalition for group query matters, but also the choice of the right sequence of these queries. Thus, $\ourlogic$ allows us to reason 
about planning the order of queries as well.




\section{A Sound and Complete Axiomatisation of $\ourlogic$}


The axiomatisation of $\ourlogic$ is presented in \Cref{tab:ax}. Axioms (I1)-(I6) for linear inequalities were proposed in \cite{HalpernInequalities}. Axioms (K)-(C), (Nec$_i$), (RC1) are standard axioms and inference rules for epistemic logic with common knowledge $S5^C$ \cite{HalpernBook}. (B$^+$)-(c$^\approx$) are axioms ensuring all necessary properties of  $\Bdg$ and $\Cost$ functions. Axioms (r$_p$)--(r$_{K2}$), (Rep), (RC2) are reduction-style axioms and rules for dynamic operator $[?_G^A]$. In $r_{\geq}$ axiom, $\bigl(\sum\limits_{i=1}^{k}a_it_i\geq z\bigr)^{(G, A)}$ denotes $\bigl(\sum\limits_{i=1}^{k}a_it_i\geq z\bigr)$, in which all occurrences of $b_i$ for $i\in G$ among $t_1, \dots, t_k$ are replaced with $(b_i-\frac{c_j(A)}{|G|})$, where $j$ is the agent that occurs in $\text{BCS}(G, A)$, i.e. for whom the cost of $A$ is minimal:  $\bigwedge\limits_{k\in G}c_j(A)\leq c_k(a)$.
Soundness of (r$_p$)--(r$_{K2}$) can be shown by a direct application of the definition of semantics. For details see Technical Appendix. 
\begin{theorem}
    The axiomatisation of $\ourlogic$ is sound.
\end{theorem}

If there is a derivation of $\varphi$ from the axioms and rules of inference of $\ourlogic$, we say that $\varphi$ is a \emph{theorem} of $\ourlogic$ and write $\vdash_{\ourlogic} \varphi$. We write $\vdash$, when the logic we refer to is clear from the context.

Observe that the presence of reduction axioms for the dynamic operator allows us to `translate away' the dynamic operator thus showing that $\ourlogic$ \textit{without common knowledge} is equally expressive as $\ourlogic$ \textit{without common knowledge and dynamic operators}. The completeness of $\ourlogic$ without common knowledge follows trivially from the completeness of $\ourlogic$ without common knowledge and dynamic operators \cite{HalpernInequalities}.

Note that we do not have a complete reduction of $\ourlogic$ to its static fragment due to the presence of common knowledge \cite{DELbook,Wang2013}.
Indeed, we can show that $\ourlogic$ is strictly more expressive than $\ourlogic$ without common knowledge by reusing the argument from, e.g., \cite[Theorem 8.34]{DELbook} and setting all agents' budgets to 0, and cost of all propositional formulas to 1.
So, the reduction argument does not work for our completeness proof.

\begin{table}[t]
    \centering
\begin{tabular}{||ll||}
\hline 
\textbf{Axioms:} &\\
(Taut)  & All propositional tautologies \\
(I1)& $(\sum\limits_{i=1}^{k}a_it_i\geq c) \leftrightarrow (\sum\limits_{i=1}^{k}a_it_i+0t_{k+1}\geq c)$\\
(I2)& $(\sum\limits_{i=1}^{k}a_it_i\geq c) \rightarrow (\sum\limits_{i=1}^{k} a_{j_i}t_{j_i}\geq c)$,\\ 
&
where $j_1,\dots, j_k$ is a permutation of $1, \dots, k$\\
(I3)& $(\sum\limits_{i=1}^{k}a_it_i\geq c)\wedge (\sum\limits_{i=1}^{k}a'_it_i\geq c')\to$\\ & $\to \sum\limits_{i=1}^{k}(a_i+a'_i)t_i\geq (c+c')$\\
(I4)& $(\sum\limits_{i=1}^{k}a_it_i\geq c)\leftrightarrow (\sum\limits_{i=1}^{k}da_it_i\geq dc)$ for $d>0$\\
(I5)& $(t\geq c)\vee (t\leq c)$\\
(I6)& $(t\geq c)\to (t>d)$, where $c>d$\\
    \hline
        (K)& $K_i(\varphi \to \psi)\to (K_i\varphi \to K_i \psi)$\\
(T)& $K_i\varphi \to \varphi$\\
(4)& $K_i\varphi \to K_iK_i\varphi$\\
(5)& $\neg K_i\varphi \to K_i\neg K_i\varphi$\\
(C)& $C_G \varphi \to E_G(\varphi \wedge C_G\varphi)$ \\
\hline
(B$^+$)& $b_i \ge 0 $\\
(c$^+$)& $c_i(A) \ge 0 $\\
(c$^\top$) & $c_i(\top) = 0$\\
(c$^\approx$)&$c_i(A) = c_i(B)$ if $A\approx B$\\
\hline
(r$_{p}$) & $[?_G^A]p \leftrightarrow (\text{BCS}(G, A)\to p)$ \\
(r$_{\ge}$) & $[?_G^A]\bigl(\sum\limits_{i=1}^{k}a_it_i\geq z\bigr) \leftrightarrow $\\
&$\leftrightarrow \Big(\text{BCS}(G, A)\to \bigl(\sum\limits_{i=1}^{k}a_it_i\geq z\bigr)^{(G, A)}\Big)$\\
(r$_{\neg}$) & $[?_G^A]\neg \varphi \leftrightarrow \text{BCS}(G, A) \to \neg [?_G^A]\varphi$ \\
(r$_{\wedge}$) & $[?_G^A](\varphi \wedge \psi)   \leftrightarrow [?_G^A]\varphi \wedge [?_G^A]\psi$\\
(r$_{K1}$) & $[?_G^A]K_j \varphi \leftrightarrow \text{BCS}(G, A) \to K_j [?_G^A]\varphi$, for $j\notin G$ \\
(r$_{K2}$) & $[?_G^A]K_i \varphi \leftrightarrow \text{BCS}(G, A) \to$\\
&$\to\bigwedge\limits_{A'\in \{A, \neg A\}}\Big(\bigl(A' \to K_i(A' \to [?_G^A]\varphi)\bigr)\Big)$, for $i\in G$\\
\hline
\textbf{Rules:} & \\
(MP) & from $\varphi$ and $\varphi\rightarrow\psi$, infer $ \psi$ \\
(Nec$_{i}$) & from $\varphi$ infer $ K_i\varphi$\\
(RC1)& from $ \varphi \to E_G(\varphi \wedge \psi)$, infer $ \varphi \to C_G\psi$\\
(Rep)& from $\varphi \leftrightarrow \psi$, infer \\
&$[?_G^A]\varphi \leftrightarrow [?_G^A]\psi$ for any $A\in\cL_{PL}$\\
(RC2) & from $\chi\to[?_G^A]\psi$ and $(\chi \wedge \text{BCS}(G, A))\to$\\
&$\to\bigwedge\limits_{A'\in \{A, \neg A\}}(A'\to E_{H\cap G}(A'\to\chi))\wedge E_{H\setminus G}\chi$,  \\
& infer $\chi\to [?_G^A]C_H\psi$ \\
         \hline
    \end{tabular}
    \caption{Proof system for $\ourlogic$.}
    \label{tab:ax}
\end{table}

Moreover, $\ourlogic$ is not compact due to the presence of common knowledge. But even the  $C_G$-free fragment is not compact due to the linear inequalities. Consider a set of $\ourlogic$-formulas: $\{c_i(A)>n\mid n\in \mathbb{N}\}$. It is easy to see that any finite subset of this set is satisfiable while the set itself is not. 

\begin{theorem}
    $\ourlogic$ is not compact.
\end{theorem}

For this reason, in the rest of this section we prove the \textit{weak completeness} of  $\ourlogic$. The proof is organised as follows. First, we define a Fisher-Ladner \cite{Fischer_Ladner_1977} style closure $cl(\varphi)$ for any $\ourlogic$-consistent formula $\varphi$. Then we construct a finite canonical pre-model, in which $\Cost$ and $\Bdg$ functions are undefined and prove that such functions satisfying (C1) and (C2) exist. It gives us a finite canonical model, for which we prove the truth lemma and establish completeness. Later, we use the bounded size of canonical model to prove that $\ourlogic$ is decidable.

\begin{definition}[Closure] \label{def:closure}
Let $cl(\varphi)$ be the smallest set of formulas such that
\begin{enumerate}
\item $cl(\varphi)$  contains $\Sub(\varphi)$, i.e. all subformulas of $\varphi$;
\item $cl(\varphi)$ is closed under single negation: if $\psi\in cl(\varphi)$
and $\psi$ does not start with $\neg$, then ${\neg\psi\in cl(\varphi)}$;
\item $(b_i\ge 0)\in cl(\varphi)$, for each agent $i\in\AG$;
\item $(c_i(A)\ge 0)\in cl(\varphi)$, for each $i\in \AG$ and each $A \in cl(\varphi)$; 
\item $c_i(\top) = 0\in cl(\varphi)$,
\item $c_i(A) = c_i(B)\in cl(\varphi)$ for all $A,B \in cl(\varphi)$ s.t. $A\approx B$;
\item if $C_G\psi \in cl(\varphi)$, then $E_G(\psi \wedge C_G\psi)\in cl(\varphi)$;
\item if $[?_G^A]p \in cl(\varphi)$, then $\text{BCS}(G, A)\to p \in cl(\varphi)$;
\item if $[?_G^A]\bigl(\sum\limits_{i=1}^{k}a_it_i\geq z\bigr)\in cl(\varphi)$, then \\
$\Big(\text{BCS}(G, A)\to \bigl(\sum\limits_{i=1}^{k}a_it_i\geq z\bigr)^{(G, A)}\Big)\in cl(\varphi)$;
\item if $[?_G^A]\neg \psi \in cl(\varphi)$, then $(\text{BCS}(G, A) \to \neg [?_G^A]\psi) \in cl(\varphi)$;
\item if $[?_G^A](\chi \wedge \psi) \in cl(\varphi)$, then $[?_G^A]\chi \wedge [?_G^A]\psi \in cl(\varphi)$;
\item if $[?_G^A]K_j \psi\in cl(\varphi)$, where $j\notin G$, then\\ $\text{BCS}(G, A) \to K_j[?_G^A]\psi\in cl(\varphi)$;
\item if $[?_G^A]K_i\psi\in cl(\varphi)$, where $i\in G$, then $\text{BCS}(G, A) \to$ \\
$\to \bigwedge\limits_{A'\in \{A, \neg A\}}\Bigl(\bigl(A' \to K_i(A' \to [?_G^A]\psi)\bigl)\Bigl) \in cl(\varphi)$;
\item if $[?_G^A]C_G \psi \in cl(\varphi)$, then $\{[?_G^A]K_i C_G \psi \mid i \in G\} \subseteq cl(\varphi)$. %
\end{enumerate}
\end{definition}

This construction guarantees that $cl(\varphi)$ is non-empty and finite as well as the set of all maximal
consistent subsets of $cl(\varphi)$. Note also that $|cl(\varphi)|$ is polynomial in $|\varphi|$, where $|\varphi|=|\Sub(\varphi)|$.

Now we construct a finite \emph{canonical pre-model} 
 which does not (yet) contain $\Cost$ and $\Bdg$ functions.

\begin{definition}[Canonical pre-model]\label{def:premodel} Given a $\ourlogic$-consistent formula $\varphi$, let a \emph{canonical pre-model} for $\varphi$ be a tuple \[\cM^{cp} = (W^c, \{\sim^c_i\}_{i\in\AG}, V^c), \text{where}\]
\begin{itemize}
\item $W^c$ is the set of all maximal $\ourlogic$-consistent subsets of $cl(\varphi)$;
\item $x\sim^c_i y$ \ iff for all formulas $\psi\in cl(\varphi)$, it holds that
${K_i\psi\in x}$ \ iff \ ${K_i\psi\in y}$;
\item $w\in V^c(p)$ \ iff \ $p\in w$, for each $p\in cl(\varphi)$.
\end{itemize}
\end{definition}

Next we need to prove the existence of appropriate $\Cost$ and $\Bdg$ functions. 

\begin{lemma}\label{lemma:cost} There exist $\Cost^c$ and $\Bdg^c$ functions, such that for all $(\sum\limits_{i=1}^{k}a_it_i \ge z)$ and all $w\in W^c$ it holds that $(\sum\limits_{i=1}^{k}a_it_i \ge z)\in w$ iff $\sum\limits_{i=1}^{k}a_it'_i \ge z$, where
\[t'_k =
  \begin{cases}
    \Cost_i^c(w, A), & \text{for } t_k = c_i{(A)} \\
    \Bdg_i^c(w), & \text{for } t_k = b_i 
  \end{cases}\]
\end{lemma}
\begin{proof}
Since every state $w\in W^c$ is $\ourlogic$-consistent,
the set of all linear inequalities contained in $w$ is satisfiable,
i.e. has at least one solution, due to (I1)-(I6) \cite{HalpernInequalities}. Then we can construct functions $\Cost^c_i(w, A)$ and $\Bdg^c_i(w)$ that agree with this solution:
for formulas $A\in\LPL$ such that $c_i(A)$ occurs in $cl(\varphi)$,
we put $\Cost^c_i(w, A)$ to be the rational that corresponds
to $c_i(A)$ in that solution;
for other formulas $B\in\LPL$, if $B\approx A$ for some formula $A\in cl(\varphi)$, then we put $\Cost^c_i(w, B):=\Cost_i^c(w, A)$. Thus we can enforce that for all $w\in W^c$ and all $A\in \LPL$ such that $c_i(A) \in cl(\varphi)$ it holds that \\(1) $\Cost^c_i(w, A)\ge 0$ for all $i\in \AG$ and all formulas $A\in\LPL$ such that $A$ occurs in $cl(\varphi)$ by the construction of $cl(\varphi)$ and $(c^+)$ axiom,
\\(2) $\Cost_i^c(w, \top)=0$, by the construction of $cl(\varphi)$ and $(c^\top)$ axiom,
\\(3) $\Cost_i^c(w, A)=\Cost_i^c(w, B)$ for all $A,B\in\LPL$ such that $A\approx B$, by $(c^\approx)$ axiom.

Similarly, we construct $\Bdg^c_i(w)$ functions in accordance with the existing solution of linear inequalities contained in $w$. This construction is well-defined, and for any $w\in W^c$ and any $i\in \AG$,
it holds that $\Bdg_i^c(w)\ge 0$ by $(B^+)$ axiom and the construction of $cl(\varphi)$.
\end{proof}

This finishes the construction of a finite \textit{canonical model} $\cM^c = (W^c\!, (\sim^c_i)_{i\in \AG}$, $\Cost^c\!, \Bdg^c\!, V^c)$.
As we have already demonstrated, this model satisfies  (C1) and (C2).
It is also clear that for all $i\in \AG$, ${\sim_i^c}$
is an equivalence relation on~$W^c$ (for details see \cite{HALPERN1992}). Moreover, all $w \in W^c$ are deductively closed in $cl(\varphi)$ \cite[Lemma 7.31]{DELbook}.




\begin{lemma}[Truth Lemma]\label{lemma:truth} Let $\cM^c$ be the canonical model for $\varphi$. For all $\psi\in cl(\varphi), w \in W^c:  \cM^c, w \vDash \psi$ iff $\psi\in w$.
\end{lemma}
\begin{proof}

It is relatively straightforward to define a \emph{complexity} measure $c$ for formulas
$\varphi,\psi \in \ourlogic$ such that if $\varphi$ is one of the antecedents of (r$_p$)--(r$_{K2}$) and $\psi$ is a corresponding consequent, then $c(\varphi) > c(\psi)$. For details see \cite{DELbook} and Technical Appendix. 

\textbf{Induction Hypothesis (IH)}. For all $c(\psi) < c(\varphi)$ and all maximal consistent subsets $w$ of $cl(\varphi)$, $\cM^c, w \models \psi$ if and only if $\psi \in w$.

\textbf{Cases} for $\psi = p, \neg, \wedge$ and $K_i\varphi'$ are trivial.

\textbf{Case} for $\psi=(z_1t_1+\dots + z_nt_n) \ge z$ follows straightforwardly from the choice of $\Cost^c$ and $\Bdg^c$ in \Cref{lemma:cost}. 

\textbf{Case} $\psi = C_G\varphi'$. We follow the proof from \cite{HALPERN1992}. 

\textit{Right-to-left}. Assume $C_G\varphi' \in w$. We will show by induction on $k$ that if $w'$ is $\sim_G$-reachable from $w$ in $k$ steps then $\varphi' \in w'$ and $C_G\varphi' \in w'$. For the case of $k=1$ it is clear that if $C_G\varphi' \in w'$ then $E_G(\varphi' \wedge C_G\varphi')\in w$ by the construction of $cl(\varphi)$ and the fact that $w$ is its maximally consistent subset. Then for all $y\in W^c$ if $y$ is $\sim_G$-reachable from $w$ in one step, then both $\varphi' \in y$ and $C_G\varphi'\in y$ since for some $i\in \AG$ it holds that $\forall y: w\sim_iy$ and $E_G(\varphi' \wedge C_G\varphi')\in w \Rightarrow (\varphi' \wedge C_G\varphi')\in y$. Now we can prove the induction step: assume our statement holds for $k$ and prove that it also holds for $k+1$. Assume that $w'\in W^c$ is $\sim_G$-reachable from $w$ in $k+1$ steps. Then exists $t\in W^c$ which is $\sim_G$-reachable from $w$ in $k$ steps and $w'$ is $G$-reachable from $t$ in one step. By our induction hypothesis, $C_G\varphi' \in t$ and $\varphi' \in t$. By our first argument it is clear that $\varphi' \in w'$. Then by our main induction hypothesis $M, w'\vDash \varphi$ for all $w'$ which are $\sim_G$-reachable from $w$. Then $M, w\vDash C_G\varphi'$.

\textit{Left-to-right}. Assume $M, w\vDash C_G\varphi'$. Note that
 every $y\in W^c$ contains a finite set of formulas. Then we
 can write their conjunction in our language: $\varphi_y$.  
 Let $ \chi = \bigvee\limits_{y\in \{w\mid M, w\vDash C_G\varphi'\}} \varphi_y$. Now it is easy to see that the following statements hold: $\vdash_\ourlogic \varphi_ w\to \chi$, $\vdash_\ourlogic \chi \to \varphi'$, 
$\vdash_\ourlogic \chi \to E_G\chi$.
It follows straightforwardly that $\vdash \varphi_w \to C_G\varphi'$, and hence $C_G\varphi' \in w$. Otherwise it would hold that $\neg C_G\varphi' \in w$ which would imply inconsistency of $w$.

\textbf{Case} $\psi = [?_G^A]\varphi'$. \textbf{Subcases} for $\psi = [?_G^A]p$, $\psi = [?_G^A]\neg\varphi'$, $\psi = [?_G^A] \varphi'\wedge\psi$ and $[?_G^A]K_i\varphi'$ follow from the construction of $cl(\varphi)$, and the IH that $c([?_G^A]\varphi') > c(\chi)$, where $\chi$ is a consequent of one of the axioms (r$_p$)-(r$_{K2}$). 

\textbf{Subcase} $\psi = [?_G^A]C_H\varphi'$.\\ 
\textit{Right-to-left}. Suppose that $[?_G^A]C_H\varphi' \in w$ and $w \sim^{?_G^A}_H w'$. 
By definition, $w\sim^{?_G^A}_Hw'$ means that there is a finite path $w \sim_{i_1}^{?_G^A} w_1 \sim_{i_2}^{?_G^A} w_2 \sim_{i_3}^{?_G^A} \dots \sim_{i_n}^{?_G^A} w_n=w'$ such that $i_1, \dots, i_n \in H$.  We first prove that for $k \in \{1, ..., n\}$ it holds that $[?_G^A]C_H\varphi' \in w_k$ and $[?_G^A]\varphi' \in w_k$, where $w_k = w$. For this, for each $\sim_{i_k}$ we need to distinguish  cases, where $i_k \in H - G$ and 
$i_k \in H \cap G$.

First, let us consider the case $i \in H - G$ for any $i \in \{i_1, ..., i_n\}$. Assume that $\text{BCS}(G, A) \in w_{k}$. Then, from items (14) and (12) of Definition \ref{def:closure} by MP we have that $K_{i}[?_G^A]C_H\varphi' \in w$. By the construction of the canonical model and the assumption that $w_k \sim_{i}^{?_G^A} w_{k+1}$, the latter implies that $K_{i}[?_G^A]C_H\varphi' \in w_{k+1}$. Since all states of the canonical model are deductively closed, we get $[?_G^A]C_H\varphi' \in w_{k+1}$. Finally, from that the fact that $w_{k+1}$ is deductively closed, we also have $[?_G^A]\varphi' \in w_{k+1}$ due to $C_H \varphi' \to \varphi'$ and the distributivity of $[?_G^A]$\footnote{Can be shown by a straightforward application of the definition of semantics.}.


Second, we consider the case $i \in H \cap G$. 
By items (14) and (13) of Definition \ref{def:closure}, we have that   $[?^G_A]K_{i}C_H \varphi' \in w_k$ and $\bigwedge\limits_{A'\in \{A, \neg A\}}\Big(\bigl(A' \to K_{i}(A' \to [?_G^A]C_H\varphi')\bigr)\Big) \in w_k$. Without loss of generality, assume that $A \in w_k$. Then by the deductive closure of $w_k$ we obtain $K_{i}(A' \to [?_G^A]C_H\varphi')\bigr) \in w_k$. By the construction of the canonical model and the assumption that $w_k \sim_{i}^{?_G^A} w_{k+1}$, the latter implies that $A' \to [?_G^A]C_H\varphi' \in w_{k+1}$. Moreover, assumption $w_k \sim_{i}^{?_G^A} w_{k+1}$ implies that $A \in w_{k+1}$. Hence, we have $[?_G^A]C_H\varphi' \in w_{k+1}$, and similarly to the previous case, $[?_G^A]\varphi' \in w_{k+1}$.

Finally, we are ready to prove our main claim here. Recall that we assume that $[?_G^A]C_H\varphi' \in w$ and $w \sim^{?_G^A}_H w'$. We have shown that $[?_G^A]\varphi' \in w_k$ for all $w_k$ on the path $w \sim^{?_G^A}_H w'$. By IH, this is equivalent to the fact that $\cM^c, w_k \models [?_G^A]\varphi'$ for all states $w_k$ on the path $w \sim^{?_G^A}_H w'$. The latter is equivalent to $\cM^c, w \models [?_G^A]C_H\varphi'$ by the semantics.


\textit{Left-to-right}. Assume that $\cM^c, w \models [?_G^A]C_H \varphi'$, and let us define  $S:= \{ w \in W^c \mid \cM^{c}, w \models [?_G^A]C_H \varphi' \}$ and $\chi:= \bigvee \{ \underline{X} \mid X \in S\}$, where $\underline{X}:= \varphi_1 \wedge \dots \wedge \varphi_m$, such that $\{\varphi_1, \dots, \varphi_m \} = X$. Clearly, $w \in S$, and hence, due to $w$ being deductively closed, $\chi \in w$. Moreover, it is straightforward to prove that $\chi \to [?_G^A] C_H \varphi' \in w$, and hence $[?_G^A] C_H \varphi' \in w$ (see \cite[Chapter 7.5]{DELbook} for details).  
\end{proof}

\begin{theorem}[Completeness]\label{thrm:completeness} 
For any $\varphi \in \cL_{\ourlogic}$, if $\varphi$ is valid, then $\vdash \varphi$.
\end{theorem}
\begin{proof} 

Assume towards a contradiction that $\not \vdash \varphi$. This means that $\lnot \varphi$ is consistent, and thus it is in one of the maximal consistent sets $w$ built from $cl(\lnot \varphi)$. For such a closure we can construct the canonical model, such that $\cM^c, w \models \lnot \varphi$ by Lemma \ref{lemma:truth}, or, equivalently, $\cM^c, w \not \models \varphi$.
\end{proof}

As a corollary of the canonical model construction, we immediately get the small model theorem. 

\begin{theorem}[Small Model Theorem]\label{thm:smt}
    If $\varphi \in \mathcal{L}_\ourlogic$ is satisfiable, then it is satisfied in a model with at most $2^{|cl(\varphi)|}$ states.
\end{theorem}

\section{Complexity profile of $\ourlogic$}

In this section we explore the complexity of the model checking and satisfiability problems of $\ourlogic$. In particular, we first establish that model checking $\ourlogic$ can be done in polynomial time. After that, we show that $\ourlogic$ is decidable with NEXPTIME upper bound. 

\subsection{Model checking}

\begin{definition}
    Given a finite $\cM = (W, (\sim_i)_{i\in \AG}, \Cost, \Bdg, V)$ and a formula $\varphi\in \cL_{\ourlogic}$, the \emph{global model checking problem} for $\ourlogic$ consists in finding all $w \in W$ such that $\cM, w\vDash \varphi$.
\end{definition}

In this section, we provide a polynomial time algorithm for solving the global model checking problem for $\ourlogic$. The algorithm requires for a given $\varphi$ a list of subformulas of $\varphi$ ordered so that group queries are evaluated before the formulas within the scope of the dynamic modalities, i.e. in $[?_G^A]\psi$ we would like to find the extension of $A$ first, and only then the extension of $\psi$. In such a way we can simulate the effects of queries before checking the formulas that may be impacted by them. 

Let some $\varphi\in \cL_{\ourlogic}$ be given. First, we create a list $sub(\varphi)$ of subformulas of $\varphi$ that also includes modalities $[?^A_G]$ occurring in $\varphi$. After that we label each $\psi \in sub(\varphi)$ by a sequence of dynamic modalities inside the scope of which it appears. Finally, we order the list in the following way. For $\psi^\sigma$ and $\chi^\tau$ with (possibly empty) labellings $\sigma$ and $\tau$, $\psi^\sigma$ precedes $\chi^\tau$ if and only if 
\begin{itemize}
    \item $\psi^\sigma$ and $\chi^\tau$ occur in modalities $[?^A_G]$, and $\sigma < \tau$ (i.e. $\sigma$ is a proper prefix of $\tau$), or else
    \item $\psi^\sigma$ appears in some $[?^A_G]$, and $\chi^\tau$ does not, or else
    \item $\psi^\sigma$ is of the form $[?^A_G]$, and $\chi^\tau$ is not, or else
    \item neither $\psi^\sigma$ nor $\chi^\tau$ appear in any $[?^A_G]$, and $\tau < \sigma$, or else
    \item both $\psi^\sigma$ and $\chi^\tau$ are of the form $[?^A_G]$, and $\sigma < \tau$, or else
    \item $\sigma = \tau$, and $\psi^\sigma$ is a part of $\chi^\tau$, or else
    \item $\psi$ appears to left of $\chi$ in $\varphi$.
\end{itemize}

As an example, let $\varphi: = [?^p_G][?^{p \lor q}_H]C_G p$. The ordered list $sub(\varphi)$ would look as follows: 
\begin{multline*}
  \{p, [?^p_G], p^{[?^p_G]}, q^{[?^p_G]}, (p \lor q)^{[?^p_G]}, p^{[?^p_G], [?^{p \lor q}_H]},\\ (C_G p)^{[?^p_G], [?^{p \lor q}_H]}, ([?^p_G] C_G p)^{[?^{p \lor q}_H]}, \varphi\}.  
\end{multline*}
Note that the size of $sub(\varphi)$ is bounded by $\mathcal{O}(|\varphi|)$. 

Our global model checking Algorithm \ref{euclid2} for $\ourlogic$ is based on the labelling algorithm for epistemic logic (see, e.g., \cite{HALPERN1992}). 
Thus we omit all Boolean and some epistemic cases for brevity, and provide only the case of common knowledge as an example. 
The technical complexity in our algorithm is that for the case of the dynamic modalities, we should keep track of which states and relations are preserved after a sequence of updates. Moreover, we also create a polynomial number of additional budget variables to store the remaining budget of agents after each query. 

\begin{center}
  \begin{breakablealgorithm}
    \caption{An algorithm for global model checking for $\ourlogic$}\label{euclid2}
    \footnotesize
    \begin{algorithmic}[1]
      \Procedure{Global$\ourlogic$}{$M, \varphi$}
      \ForAll{$\psi^\sigma \in sub(\varphi)$}
        \ForAll{$w \in W$}
      \Case{$\psi^\sigma = C_G{\chi}^\sigma$}       
        \State{$\mathit{check} \gets \mathit{true}$}
        \ForAll{$(w,v) \in R_G$}
          \If{$(w,v)$ is labelled with $\sigma$}
          
          \If{$v$ is not labelled with $\chi^\sigma$}
            \State{$\mathit{check} \gets \mathit{false}$}
            \State{\textbf{break}}
          \EndIf
          \EndIf
        \EndFor
        \If{$\mathit{check}$}
          \State{label $w$ with $C_G{\chi}^\sigma$}
        \EndIf
      \EndCase
        
    \Case{$\psi^\sigma = [?^A_G]^\sigma$}
        \ForAll{$i \in \AG$}
            \ForAll{$(v,u)^\sigma \in \sim_i$}
                \If{$\cM,v \models BCS^\sigma (G,A)$ and $\cM,u \models BCS^\sigma(G,A)$}
                \ForAll{$j \in G$}
                    \State{$Bdg_j(v)^\sigma \gets Bdg_j^{?_G A}(v)$}
                    \State{$Bdg_j(u)^\sigma \gets Bdg_j^{?_G A}(u)$}
                \EndFor
                \If{$i \not \in G$}
                    \State{label $(v,u)$ with $\sigma, [?^A_G]$}
                    \Else
                    \If{$v$ is labelled with $A$ iff $u$ is labelled with $A$}
                        \State{label $(v,u)$ with $\sigma, [?^A_G]$}
                \EndIf
                
                    \EndIf
                \EndIf
            \EndFor
        \EndFor        
    \EndCase

    \Case{$\psi^\sigma = ([?^A_G]\chi)^\sigma$}
        \If{$w$ is labelled with $\chi^{\sigma, [?^A_G]}$}
            \State{label $w$ with $([?^A_G]\chi)^\sigma$}
        \EndIf
    \EndCase

        \EndFor
    \EndFor      
      
     \EndProcedure
    \end{algorithmic}
  \end{breakablealgorithm}
\end{center}

The algorithm mimics the definition of semantics, and its correctness can be shown by induction on $\varphi$. 
The preparation of ordered list $sub(\varphi)$ takes $\mathcal{O}(|\varphi|^2)$ number of steps. 
On line 16, $BCS^\sigma(G,A)$ is the budget constraint for query $A$ for agents from $G$ after the sequence of updates $\sigma$. Respective budgets of agents after updates are calculated on lines 18 and 19. Observe, that computing $Bdg_j(v)^\sigma$ and $BCS^\sigma (G,A)$ requires only arithmetical computation with values of all variables known. This can be done in polynomial time. Each computation of $BCS^\sigma$ is called for $\mathcal{O}(|\varphi|\cdot |W| \cdot |\AG| \cdot |\!\sim\!|)$ times. Finally, each computation of some $Bdg_j(v)^\sigma$ is called for $\mathcal{O}(|\varphi|\cdot |W| \cdot |\AG|^2 \cdot |\!\sim\!|)$ times. Since in the worst case, $BCS^\sigma$ and $Bdg_j(v)^\sigma$ require a polynomial number of steps, model checking $\ourlogic$ is in polynomial time. 

\begin{theorem}\label{thrm:modelchecking}
    Model checking $\ourlogic$ is in P.
\end{theorem}

\subsection{Decidability}
\begin{definition}
    Given a formula of $\varphi\in \cL_{\ourlogic}$, the \emph{satisfiability problem} for $\ourlogic$ consists in determining whether there is a pointed model $\cM,w$ such that $\cM,w \models \varphi$.
\end{definition}

\begin{theorem}[Decidability] The satisfiability problem for $\ourlogic$ is decidable.
\end{theorem}
\begin{proof} 
The decidability of $\ourlogic$ follows from the small model theorem (Theorem \ref{thm:smt}). This theorem states that a formula $\varphi\in \cL_{\ourlogic}$ is satisfiable iff it is satisfiable in a model
$M$ with at most $2^{|cl(\varphi)|}$ states. Usually, like in the cases of PDL \cite{Fischer_Ladner_1977} or $S5_n^C$ \cite{HALPERN1992}, there are finitely many such models, so it is sufficient to enumerate them and check whether $\varphi$ holds in any. But in our case there are infinitely many choices of $\Bdg$ and $\Cost$ functions, so there are infinitely many models with at most $2^{|cl(\varphi)|}$ states. In order to overcome this difficulty, we apply the technique used in \cite{Dragan}.  The idea is to consider pre-structures $M'=(W, (\sim_i)_{i\in \AG}, P^*, V^c)$, in which  $P^*$ is a `pseudo' function that emulates actual $\Cost$ and $\Bdg$ for all subformulas of $\varphi$ of the form $\sum\limits_{i=1}^{k}a_it_i \ge z$: 
$$P^*\colon W\times (\sum\limits_{i=1}^{k}a_it_i \ge z) \longrightarrow \{true, false\}.$$
Then we can define a satisfiability relation $\vDash'$, similarly to \Cref{def:semantics} in all cases except $(\sum\limits_{i=1}^{k}a_it_i \ge z)$. We say that $$M', w\vDash' (\sum\limits_{i=1}^{k}a_it_i \ge z) \text{ iff } P^*(w, \sum\limits_{i=1}^{k}a_it_i \ge z)=true.$$ 

Since there are only finitely many such pre-structures with $\leq 2^{|cl(\varphi)|}$ states, we can enumerate them all and check if $\varphi$ holds in any of $M'$'s according to $\vDash'$. If it does, then we need to check if $P^*$ can be replaced with ($\Cost, \Bdg$). For this purpose we define a set of linear inequalities $I(w)$ for each state $w\in M'$, such that
\begin{itemize}
    \item $(\sum\limits_{i=1}^{k}a_it_i \ge z)\in I(w)$ iff  $P^*(w, \sum\limits_{i=1}^{k}a_it_i \ge z)=true$;
    \item $(\sum\limits_{i=1}^{k}a_it_i < z)\in I(w)$ iff  $P^*(w, \sum\limits_{i=1}^{k}a_it_i \ge z)=false$;
    \item $c_i(A)\geq 0, b_i\geq 0$, $c_i(\top)=0\in I(w)$ for all $A \in cl(\varphi)\cap \LPL$, $i\in \AG$; 
    \item $c_i(A)=c_i(B)\in I(w)$ for all $A \in cl(\varphi)\cap \LPL$, s.t. $A\approx B$.
    \end{itemize}
    It remains to find at least one solution of such system of linear inequalities to define $(\Cost, \Bdg)$. But since a problem of solving a system of linear inequalities is well-known to be decidable in polynomial time \cite{KHACHIYAN1980}, then given a pre-model $M'$ we can extend it to a normal model $\cM$ according to \Cref{lemma:cost} in finitely many steps. 
    
    Finally, given a $\varphi$, we can enumerate finite pre-models of size at most $2^{|cl(\varphi)|}$, solve the corresponding systems of linear inequalities to extend these pre-models to normal models, and check whether $\varphi$ is true in any of them. If yes, the formula is satisfiable, if not, then $\varphi$ is unsatisfiable, since each satisfiable formula has a model of size at most $2^{|cl(\varphi)|}$.  




This algorithm gives us a NEXPTIME upper-bound, because each satisfiable formula $\varphi$
has a model of exponential size in $|\varphi|$.  So, we can guess an exponential model $\cM'$ and a state $w$ of the model, such that $\cM', w\vDash\varphi$, which can be checked in polynomial time by \Cref{thrm:modelchecking}.
\end{proof}

We leave finding the precise complexity bounds for future work, noting that the satisfiability problem for $\ourlogic$ is EXPTIME-hard from the EXPTIME-completeness of $S5_n^C$\cite{HALPERN1992}.

\balance

\section{Related Work}

In the literature, there are several reasons to put resource constraints on agents. One may want to deal with non-omniscient agents, and thus treat resources as limitations on their reasoning abilities \cite{FAGIN1987}. Similarly, \cite{Duc1997} explores rational but non-omniscient agents. The logics for agents as perfect reasoners  who take time to derive consequences of their knowledge were studied in \cite{Alechina2002,Alechina2004,Alechina2009}. 
In a similar vein, \cite{Balbiani2019} proposed a logic for reasoning about the formation of beliefs through perception or inference in non-omniscient resource-bounded agents.

One may also want to constrain agents' strategic abilities by introducting costs of actions. Extensions of various strategic logics, like alternating-time temporal logic \cite{alur02} and coalition logic \cite{pauly02}, for resource bounded agents were proposed and studied in \cite{Alechina_Logan_Nguen_Rakib_2010,Bulling_Farwer_2010,Alechina_Logan_Nguen_Rakib_2011,DELLAMONICA2011,Alechina_Logan_Nguyen_Raimondi_2014,Nguyen_2015,Alechina_Logan_Nguyen_Raimondi_Mostarda_2015,Naumov_ijcai2017}. 

Our work deals with knowledge and communication in the settings, where information available to
agents might be constrained by their resources. In such settings, resources would be treated as a cost of some `information mining' process for agents. 
A similar proposal introduced a logical system for reasoning about budget-constrained knowledge \cite{Naumov15}. This approach, however, deals with resource bounded knowledge statically, while we introduce a DEL-style framework. DEL-style logics with inferential actions that require spending resources were studied in \cite{Solaki2020,Solaki2022,JELIA2021}. 
Alternatively, one can also explore agents that can reason about epistemic formulas only up to a specific modal depth as well as about public announcements of bounded depth \cite{Farid_Martin2023}. Finally, logics for resource-bounded agents have also been of interest in the epistemic planning community \cite{Engesser2017,Belardinelli2021,Bolander2021}.

Finally, our work is also inspired by \cite{HalpernInequalities,Fagin_Halpern1994}, where linear inequalities were introduced to reason about probabilities. Linear inequalities were also used in a probabilistic DEL setting \cite{Achimescu_Baltag_Sack2016}. An alternative way to encode linear inequalities was proposed in \cite{DELGRANDE_Rene_Sack2019}.

\section{Discussion and Future Work}

In this paper we presented a dynamic epistemic logic for Semi-Public Queries with Budgets and Costs ($\ourlogic$) and demonstrated that this logic is complete, decidable
and has an efficient model checking procedure.  
We believe that these results can find their applications in various fields of multi-agent systems like formal verification, automated reasoning and epistemic planning.

In order to keep the generality of our framework we have tried to impose as few semantic restrictions as possible. Thus, we allow agents to be unaware of the costs of some formulas for themselves as well as for other
agents, and of their and others’ budgets. The proposed framework, however, can be straightforwardly extended to capture alternative modelling settings. Thus, one can add axioms like (A1) $(b_i=k) \to K_i(b_i=k)$ and (A2) $(c_i(A)=k)\to K_i(c_i(A)=k)$ to impose that all agents know their budget and how much it would cost them to mine formulas. 
Some existing papers on modelling resource bounded agents, e.g. \cite{Alechina_Logan_Nguen_Rakib_2011}, assume that resource bounds should be represented as vectors $(r_1, \dots, r_k)$ of costs of actions, where each $r_l\in (r_1, \dots, r_k)$ represents a specific resource. In this paper, we deal with a single resource to keep the presentation simple. However, one can also implement multiple resources in our framework. Let $c^l_i(A)$ denote the amount of $l$'s resource required from agent $i$ to make a query about $A$ and $b^l_i(A)$ denote the amount of $l$'s resource that agent $i$ has. Then, the cost of $A$ and the budget of $i$ may be represented as two vectors $(c^1_i(A), \dots, c^k_i(A))$ and $(b^1_i, \dots, b^k_i)$ respectively. Now, it is relatively straightforward to modify the proposed framework to deal with multiple resources. 

Currently, $\ourlogic$ allows only propositional formulas $A$ to occur under $[?_G^A]$. The extension of $\ourlogic$, where any formula $\varphi$ can occur in $[?_G^\varphi]$ is a matter of future work.
Apart from that, we also plan to extend our framework and allow quantification over queries in the spirit of logics of quantified announcements, e.g. APAL \cite{balbiani_baltag_ditmarsch_herzig_hoshi_delima_2008}, GAL \cite{AGOTNES2010}
, CAL \cite{Agotnes_vanDitmarsch2008,Galimullin2021}, and their versions with group knowledge \cite{gald,Ågotnes2023}. Another important direction for future work is to consider more complicated communicative actions (e.g. \cite{vanBenthem2007-VANDLF,Ciardelli2015,LPAR23:Learning_What_Others_Know}) in our settings. Finally, we would also like to find tight complexity bounds for the $\ourlogic$ satisfiability problem.




\begin{acks}
We thank the anonymous AAMAS 2024 reviewers for their incisive and constructive comments. This work (Vitaliy Dolgorukov) is an output of a research project implemented as part of the Basic Research Program at the National Research University Higher School of Economics (HSE University).
\end{acks}

\bibliographystyle{ACM-Reference-Format} 
\bibliography{sample}


\begin{thebibliography}{49}


\ifx \showCODEN    \undefined \def \showCODEN     #1{\unskip}     \fi
\ifx \showDOI      \undefined \def \showDOI       #1{#1}\fi
\ifx \showISBNx    \undefined \def \showISBNx     #1{\unskip}     \fi
\ifx \showISBNxiii \undefined \def \showISBNxiii  #1{\unskip}     \fi
\ifx \showISSN     \undefined \def \showISSN      #1{\unskip}     \fi
\ifx \showLCCN     \undefined \def \showLCCN      #1{\unskip}     \fi
\ifx \shownote     \undefined \def \shownote      #1{#1}          \fi
\ifx \showarticletitle \undefined \def \showarticletitle #1{#1}   \fi
\ifx \showURL      \undefined \def \showURL       {\relax}        \fi
\providecommand\bibfield[2]{#2}
\providecommand\bibinfo[2]{#2}
\providecommand\natexlab[1]{#1}
\providecommand\showeprint[2][]{arXiv:#2}

\bibitem[\protect\citeauthoryear{Achimescu, Baltag, and Sack}{Achimescu
  et~al\mbox{.}}{2016}]%
        {Achimescu_Baltag_Sack2016}
\bibfield{author}{\bibinfo{person}{Andreea Achimescu},
  \bibinfo{person}{Alexandru Baltag}, {and} \bibinfo{person}{Joshua Sack}.}
  \bibinfo{year}{2016}\natexlab{}.
\newblock \showarticletitle{{The probabilistic logic of communication and
  change}}.
\newblock \bibinfo{journal}{\emph{Journal of Logic and Computation}}
  \bibinfo{volume}{29}, \bibinfo{number}{7} (\bibinfo{year}{2016}),
  \bibinfo{pages}{1015--1040}.
\newblock
\showISSN{0955-792X}
\urldef\tempurl%
\url{https://doi.org/10.1093/logcom/exv084}
\showDOI{\tempurl}


\bibitem[\protect\citeauthoryear{{\AA}gotnes, Alechina, and
  Galimullin}{{\AA}gotnes et~al\mbox{.}}{2022}]%
        {gald}
\bibfield{author}{\bibinfo{person}{Thomas {\AA}gotnes},
  \bibinfo{person}{Natasha Alechina}, {and} \bibinfo{person}{Rustam
  Galimullin}.} \bibinfo{year}{2022}\natexlab{}.
\newblock \showarticletitle{Logics with Group Announcements and Distributed
  Knowledge: Completeness and Expressive Power}.
\newblock \bibinfo{journal}{\emph{Journal of Logic, Language and Information}}
  \bibinfo{volume}{31}, \bibinfo{number}{2} (\bibinfo{year}{2022}),
  \bibinfo{pages}{141--166}.
\newblock
\urldef\tempurl%
\url{https://doi.org/10.1007/S10849-022-09355-0}
\showDOI{\tempurl}


\bibitem[\protect\citeauthoryear{{\AA}gotnes, Balbiani, {van Ditmarsch}, and
  Seban}{{\AA}gotnes et~al\mbox{.}}{2010}]%
        {AGOTNES2010}
\bibfield{author}{\bibinfo{person}{Thomas {\AA}gotnes},
  \bibinfo{person}{Philippe Balbiani}, \bibinfo{person}{Hans {van Ditmarsch}},
  {and} \bibinfo{person}{Pablo Seban}.} \bibinfo{year}{2010}\natexlab{}.
\newblock \showarticletitle{Group announcement logic}.
\newblock \bibinfo{journal}{\emph{Journal of Applied Logic}}
  \bibinfo{volume}{8}, \bibinfo{number}{1} (\bibinfo{year}{2010}),
  \bibinfo{pages}{62--81}.
\newblock
\showISSN{1570-8683}
\urldef\tempurl%
\url{https://doi.org/10.1016/j.jal.2008.12.002}
\showDOI{\tempurl}


\bibitem[\protect\citeauthoryear{{\AA}gotnes and Galimullin}{{\AA}gotnes and
  Galimullin}{2023}]%
        {Ågotnes2023}
\bibfield{author}{\bibinfo{person}{Thomas {\AA}gotnes} {and}
  \bibinfo{person}{Rustam Galimullin}.} \bibinfo{year}{2023}\natexlab{}.
\newblock \showarticletitle{Quantifying over information change with common
  knowledge}.
\newblock \bibinfo{journal}{\emph{Autonomous Agents and Multi-Agent Systems}}
  \bibinfo{volume}{37}, \bibinfo{number}{1} (\bibinfo{year}{2023}),
  \bibinfo{pages}{19}.
\newblock
\showISSN{1573-7454}
\urldef\tempurl%
\url{https://doi.org/10.1007/s10458-023-09601-0}
\showDOI{\tempurl}


\bibitem[\protect\citeauthoryear{{\AA}gotnes and van Ditmarsch}{{\AA}gotnes and
  van Ditmarsch}{2008}]%
        {Agotnes_vanDitmarsch2008}
\bibfield{author}{\bibinfo{person}{Thomas {\AA}gotnes} {and}
  \bibinfo{person}{Hans van Ditmarsch}.} \bibinfo{year}{2008}\natexlab{}.
\newblock \showarticletitle{Coalitions and Announcements}. In
  \bibinfo{booktitle}{\emph{Proceedings of the 7th {AAMAS}}},
  \bibfield{editor}{\bibinfo{person}{Lin Padgham}, \bibinfo{person}{David~C.
  Parkes}, \bibinfo{person}{J{\"{o}}rg~P. M{\"{u}}ller}, {and}
  \bibinfo{person}{Simon Parsons}} (Eds.). \bibinfo{publisher}{{IFAAMAS}},
  \bibinfo{pages}{673–680}.
\newblock
\showISBNx{9780981738116}


\bibitem[\protect\citeauthoryear{Alechina and Logan}{Alechina and
  Logan}{2002}]%
        {Alechina2002}
\bibfield{author}{\bibinfo{person}{Natasha Alechina} {and}
  \bibinfo{person}{Brian Logan}.} \bibinfo{year}{2002}\natexlab{}.
\newblock \showarticletitle{Ascribing Beliefs to Resource Bounded Agents}. In
  \bibinfo{booktitle}{\emph{Proceedings of the 1st {AAMAS}}},
  \bibfield{editor}{\bibinfo{person}{Cristiano Castelfranchi},
  \bibinfo{person}{Maria Gini}, \bibinfo{person}{Toru Ishida}, {and}
  \bibinfo{person}{W.~Lewis Johnson}} (Eds.). \bibinfo{publisher}{ACM},
  \bibinfo{pages}{881–888}.
\newblock
\showISBNx{1581134800}
\urldef\tempurl%
\url{https://doi.org/10.1145/544862.544948}
\showDOI{\tempurl}


\bibitem[\protect\citeauthoryear{Alechina and Logan}{Alechina and
  Logan}{2009}]%
        {Alechina2009}
\bibfield{author}{\bibinfo{person}{Natasha Alechina} {and}
  \bibinfo{person}{Brian Logan}.} \bibinfo{year}{2009}\natexlab{}.
\newblock \showarticletitle{A Logic of Situated Resource-Bounded Agents}.
\newblock \bibinfo{journal}{\emph{Journal of Logic, Language and Information}}
  \bibinfo{volume}{18}, \bibinfo{number}{1} (\bibinfo{year}{2009}),
  \bibinfo{pages}{79--95}.
\newblock
\urldef\tempurl%
\url{https://doi.org/10.1007/s10849-008-9073-6}
\showDOI{\tempurl}


\bibitem[\protect\citeauthoryear{Alechina, Logan, Nga, and Rakib}{Alechina
  et~al\mbox{.}}{2010}]%
        {Alechina_Logan_Nguen_Rakib_2010}
\bibfield{author}{\bibinfo{person}{Natasha Alechina}, \bibinfo{person}{Brian
  Logan}, \bibinfo{person}{Nguyen~Hoang Nga}, {and} \bibinfo{person}{Abdur
  Rakib}.} \bibinfo{year}{2010}\natexlab{}.
\newblock \showarticletitle{Resource-Bounded Alternating-Time Temporal Logic}.
  In \bibinfo{booktitle}{\emph{Proceedings of the 9th {AAMAS}}},
  \bibfield{editor}{\bibinfo{person}{Wiebe van~der Hoek},
  \bibinfo{person}{Gal~A. Kaminka}, \bibinfo{person}{Yves Lesp{\'{e}}rance},
  \bibinfo{person}{Michael Luck}, {and} \bibinfo{person}{Sandip Sen}} (Eds.).
  \bibinfo{publisher}{IFAAMAS}, \bibinfo{pages}{481–488}.
\newblock


\bibitem[\protect\citeauthoryear{Alechina, Logan, Nga~Nguyen, and
  Rakib}{Alechina et~al\mbox{.}}{2011}]%
        {Alechina_Logan_Nguen_Rakib_2011}
\bibfield{author}{\bibinfo{person}{Natasha Alechina}, \bibinfo{person}{Brian
  Logan}, \bibinfo{person}{Hoang Nga~Nguyen}, {and} \bibinfo{person}{Abdur
  Rakib}.} \bibinfo{year}{2011}\natexlab{}.
\newblock \showarticletitle{Logic for coalitions with bounded resources}.
\newblock \bibinfo{journal}{\emph{Journal of Logic and Computation}}
  \bibinfo{volume}{21}, \bibinfo{number}{6} (\bibinfo{year}{2011}),
  \bibinfo{pages}{907--937}.
\newblock
\urldef\tempurl%
\url{https://doi.org/10.1093/logcom/exq032}
\showDOI{\tempurl}


\bibitem[\protect\citeauthoryear{Alechina, Logan, Nguyen, and
  Raimondi}{Alechina et~al\mbox{.}}{2014}]%
        {Alechina_Logan_Nguyen_Raimondi_2014}
\bibfield{author}{\bibinfo{person}{Natasha Alechina}, \bibinfo{person}{Brian
  Logan}, \bibinfo{person}{Hoang~Nga Nguyen}, {and} \bibinfo{person}{Franco
  Raimondi}.} \bibinfo{year}{2014}\natexlab{}.
\newblock \showarticletitle{Decidable Model-Checking for a Resource Logic with
  Production of Resources}. In \bibinfo{booktitle}{\emph{Proceedings of the
  21st {ECAI}}}, \bibfield{editor}{\bibinfo{person}{Torsten Schaub},
  \bibinfo{person}{Gerhard Friedrich}, {and} \bibinfo{person}{Barry
  O'Sullivan}} (Eds.). \bibinfo{publisher}{IOS Press}, \bibinfo{pages}{9–14}.
\newblock
\urldef\tempurl%
\url{https://doi.org/10.3233/978-1-61499-419-0-9}
\showDOI{\tempurl}


\bibitem[\protect\citeauthoryear{Alechina, Logan, Nguyen, Raimondi, and
  Mostarda}{Alechina et~al\mbox{.}}{2015}]%
        {Alechina_Logan_Nguyen_Raimondi_Mostarda_2015}
\bibfield{author}{\bibinfo{person}{Natasha Alechina}, \bibinfo{person}{Brian
  Logan}, \bibinfo{person}{Hoang~Nga Nguyen}, \bibinfo{person}{Franco
  Raimondi}, {and} \bibinfo{person}{Leonardo Mostarda}.}
  \bibinfo{year}{2015}\natexlab{}.
\newblock \showarticletitle{Symbolic Model-Checking for Resource-Bounded ATL}.
  In \bibinfo{booktitle}{\emph{Proceedings of the 14th {AAMAS}}},
  \bibfield{editor}{\bibinfo{person}{Gerhard Weiss}, \bibinfo{person}{Pinar
  Yolum}, \bibinfo{person}{Rafael~H. Bordini}, {and} \bibinfo{person}{Edith
  Elkind}} (Eds.). \bibinfo{publisher}{IFAAMAS}, \bibinfo{pages}{1809–1810}.
\newblock


\bibitem[\protect\citeauthoryear{Alechina, Logan, and Whitsey}{Alechina
  et~al\mbox{.}}{2004}]%
        {Alechina2004}
\bibfield{author}{\bibinfo{person}{Natasha Alechina}, \bibinfo{person}{Brian
  Logan}, {and} \bibinfo{person}{Mark Whitsey}.}
  \bibinfo{year}{2004}\natexlab{}.
\newblock \showarticletitle{A Complete and Decidable Logic for Resource-Bounded
  Agents}. In \bibinfo{booktitle}{\emph{Proceedings of the 3rd {AAMAS}}}.
  \bibinfo{publisher}{{IEEE} Computer Society}, \bibinfo{pages}{606--613}.
\newblock
\urldef\tempurl%
\url{https://doi.org/10.1109/AAMAS.2004.10090}
\showDOI{\tempurl}


\bibitem[\protect\citeauthoryear{Alur, Henzinger, and Kupferman}{Alur
  et~al\mbox{.}}{2002}]%
        {alur02}
\bibfield{author}{\bibinfo{person}{Rajeev Alur}, \bibinfo{person}{Thomas~A.
  Henzinger}, {and} \bibinfo{person}{Orna Kupferman}.}
  \bibinfo{year}{2002}\natexlab{}.
\newblock \showarticletitle{Alternating-time temporal logic}.
\newblock \bibinfo{journal}{\emph{Journal of the {ACM}}} \bibinfo{volume}{49},
  \bibinfo{number}{5} (\bibinfo{year}{2002}), \bibinfo{pages}{672--713}.
\newblock
\urldef\tempurl%
\url{https://doi.org/10.1145/585265.585270}
\showDOI{\tempurl}


\bibitem[\protect\citeauthoryear{Arthaud and Rinard}{Arthaud and
  Rinard}{2023}]%
        {Farid_Martin2023}
\bibfield{author}{\bibinfo{person}{Farid Arthaud} {and}
  \bibinfo{person}{Martin~C. Rinard}.} \bibinfo{year}{2023}\natexlab{}.
\newblock \showarticletitle{Depth-bounded Epistemic Logic}. In
  \bibinfo{booktitle}{\emph{Proceedings of the 19th {TARK}}}
  \emph{(\bibinfo{series}{{EPTCS}}, Vol.~\bibinfo{volume}{379})},
  \bibfield{editor}{\bibinfo{person}{Rineke Verbrugge}} (Ed.).
  \bibinfo{pages}{46--65}.
\newblock
\urldef\tempurl%
\url{https://doi.org/10.4204/EPTCS.379.7}
\showDOI{\tempurl}


\bibitem[\protect\citeauthoryear{Balbiani, Baltag, Ditmarsch, Herzig, Hoshi,
  and de~Lima}{Balbiani et~al\mbox{.}}{2008}]%
        {balbiani_baltag_ditmarsch_herzig_hoshi_delima_2008}
\bibfield{author}{\bibinfo{person}{Philippe Balbiani},
  \bibinfo{person}{Alexandru Baltag}, \bibinfo{person}{Hans~van Ditmarsch},
  \bibinfo{person}{Andreas Herzig}, \bibinfo{person}{Tomohiro Hoshi}, {and}
  \bibinfo{person}{Tiago de Lima}.} \bibinfo{year}{2008}\natexlab{}.
\newblock \showarticletitle{‘Knowable’ as ‘Known After an
  Announcement’}.
\newblock \bibinfo{journal}{\emph{The Review of Symbolic Logic}}
  \bibinfo{volume}{1}, \bibinfo{number}{3} (\bibinfo{year}{2008}),
  \bibinfo{pages}{305–334}.
\newblock
\urldef\tempurl%
\url{https://doi.org/10.1017/S1755020308080210}
\showDOI{\tempurl}


\bibitem[\protect\citeauthoryear{Balbiani, Fern{\'a}ndez-Duque, and
  Lorini}{Balbiani et~al\mbox{.}}{2019}]%
        {Balbiani2019}
\bibfield{author}{\bibinfo{person}{Philippe Balbiani}, \bibinfo{person}{David
  Fern{\'a}ndez-Duque}, {and} \bibinfo{person}{Emiliano Lorini}.}
  \bibinfo{year}{2019}\natexlab{}.
\newblock \showarticletitle{The Dynamics of Epistemic Attitudes in
  Resource-Bounded Agents}.
\newblock \bibinfo{journal}{\emph{Studia Logica}} \bibinfo{volume}{107},
  \bibinfo{number}{3} (\bibinfo{year}{2019}), \bibinfo{pages}{457--488}.
\newblock
\urldef\tempurl%
\url{https://doi.org/10.1007/s11225-018-9798-4}
\showDOI{\tempurl}


\bibitem[\protect\citeauthoryear{Baltag and Smets}{Baltag and Smets}{2020}]%
        {LPAR23:Learning_What_Others_Know}
\bibfield{author}{\bibinfo{person}{Alexandru Baltag} {and}
  \bibinfo{person}{Sonja Smets}.} \bibinfo{year}{2020}\natexlab{}.
\newblock \showarticletitle{Learning What Others Know}. In
  \bibinfo{booktitle}{\emph{Proceedings of the 23rd {LPAR}}}
  \emph{(\bibinfo{series}{EPiC Series in Computing},
  Vol.~\bibinfo{volume}{73})}, \bibfield{editor}{\bibinfo{person}{Elvira
  Albert} {and} \bibinfo{person}{Laura Kovacs}} (Eds.).
  \bibinfo{publisher}{EasyChair}, \bibinfo{pages}{90--119}.
\newblock
\showISSN{2398-7340}
\urldef\tempurl%
\url{https://doi.org/10.29007/plm4}
\showDOI{\tempurl}


\bibitem[\protect\citeauthoryear{Belardinelli and Rendsvig}{Belardinelli and
  Rendsvig}{2021}]%
        {Belardinelli2021}
\bibfield{author}{\bibinfo{person}{Gaia Belardinelli} {and}
  \bibinfo{person}{Rasmus~K. Rendsvig}.} \bibinfo{year}{2021}\natexlab{}.
\newblock \showarticletitle{Epistemic Planning with Attention as a Bounded
  Resource}. In \bibinfo{booktitle}{\emph{Proceedings of the 8th {LORI}}}
  \emph{(\bibinfo{series}{LNCS}, Vol.~\bibinfo{volume}{13039})},
  \bibfield{editor}{\bibinfo{person}{Sujata Ghosh} {and}
  \bibinfo{person}{Thomas Icard}} (Eds.). \bibinfo{publisher}{Springer},
  \bibinfo{pages}{14--30}.
\newblock


\bibitem[\protect\citeauthoryear{Bolander, Dissing, and Herrmann}{Bolander
  et~al\mbox{.}}{2021}]%
        {Bolander2021}
\bibfield{author}{\bibinfo{person}{Thomas Bolander}, \bibinfo{person}{Lasse
  Dissing}, {and} \bibinfo{person}{Nicolai Herrmann}.}
  \bibinfo{year}{2021}\natexlab{}.
\newblock \showarticletitle{{DEL-based Epistemic Planning for Human-Robot
  Collaboration: Theory and Implementation}}. In
  \bibinfo{booktitle}{\emph{Proceedings of the 18th {KR}}},
  \bibfield{editor}{\bibinfo{person}{Meghyn Bienvenu}, \bibinfo{person}{Gerhard
  Lakemeyer}, {and} \bibinfo{person}{Esra Erdem}} (Eds.).
  \bibinfo{pages}{120--129}.
\newblock
\urldef\tempurl%
\url{https://doi.org/10.24963/kr.2021/12}
\showDOI{\tempurl}


\bibitem[\protect\citeauthoryear{Bulling and Farwer}{Bulling and
  Farwer}{2010}]%
        {Bulling_Farwer_2010}
\bibfield{author}{\bibinfo{person}{Nils Bulling} {and} \bibinfo{person}{Berndt
  Farwer}.} \bibinfo{year}{2010}\natexlab{}.
\newblock \showarticletitle{On the (Un-)Decidability of Model Checking
  Resource-Bounded Agents}. In \bibinfo{booktitle}{\emph{Proceedings of the
  19th {ECAI}}} \emph{(\bibinfo{series}{Frontiers in Artificial Intelligence
  and Applications}, Vol.~\bibinfo{volume}{215})},
  \bibfield{editor}{\bibinfo{person}{Helder Coelho}, \bibinfo{person}{Rudi
  Studer}, {and} \bibinfo{person}{Michael~J. Wooldridge}} (Eds.).
  \bibinfo{publisher}{{IOS} Press}, \bibinfo{pages}{567--572}.
\newblock
\urldef\tempurl%
\url{https://doi.org/10.3233/978-1-60750-606-5-567}
\showDOI{\tempurl}


\bibitem[\protect\citeauthoryear{Cao and Naumov}{Cao and Naumov}{2017}]%
        {Naumov_ijcai2017}
\bibfield{author}{\bibinfo{person}{Rui Cao} {and} \bibinfo{person}{Pavel
  Naumov}.} \bibinfo{year}{2017}\natexlab{}.
\newblock \showarticletitle{Budget-Constrained Dynamics in Multiagent Systems}.
  In \bibinfo{booktitle}{\emph{Proceedings of the 26th {IJCAI}}},
  \bibfield{editor}{\bibinfo{person}{Carles Sierra}} (Ed.).
  \bibinfo{pages}{915--921}.
\newblock
\urldef\tempurl%
\url{https://doi.org/10.24963/ijcai.2017/127}
\showDOI{\tempurl}


\bibitem[\protect\citeauthoryear{Ciardelli and Roelofsen}{Ciardelli and
  Roelofsen}{2015}]%
        {Ciardelli2015}
\bibfield{author}{\bibinfo{person}{Ivano~A. Ciardelli} {and}
  \bibinfo{person}{Floris Roelofsen}.} \bibinfo{year}{2015}\natexlab{}.
\newblock \showarticletitle{Inquisitive dynamic epistemic logic}.
\newblock \bibinfo{journal}{\emph{Synthese}} \bibinfo{volume}{192},
  \bibinfo{number}{6} (\bibinfo{year}{2015}), \bibinfo{pages}{1643--1687}.
\newblock
\showISSN{1573-0964}
\urldef\tempurl%
\url{https://doi.org/10.1007/s11229-014-0404-7}
\showDOI{\tempurl}


\bibitem[\protect\citeauthoryear{Costantini, Formisano, and Pitoni}{Costantini
  et~al\mbox{.}}{2021}]%
        {JELIA2021}
\bibfield{author}{\bibinfo{person}{Stefania Costantini},
  \bibinfo{person}{Andrea Formisano}, {and} \bibinfo{person}{Valentina
  Pitoni}.} \bibinfo{year}{2021}\natexlab{}.
\newblock \showarticletitle{An Epistemic Logic for Multi-agent Systems with
  Budget and Costs}. In \bibinfo{booktitle}{\emph{Proceedings of the 17th
  {JELIA}}} \emph{(\bibinfo{series}{LNCS}, Vol.~\bibinfo{volume}{12678})},
  \bibfield{editor}{\bibinfo{person}{Wolfgang Faber}, \bibinfo{person}{Gerhard
  Friedrich}, \bibinfo{person}{Martin Gebser}, {and} \bibinfo{person}{Michael
  Morak}} (Eds.). \bibinfo{publisher}{Springer}, \bibinfo{pages}{101--115}.
\newblock
\urldef\tempurl%
\url{https://doi.org/10.1007/978-3-030-75775-5\_8}
\showDOI{\tempurl}


\bibitem[\protect\citeauthoryear{Dautovi{\'{c}}, Doder, and
  Ognjanovi{\'{c}}}{Dautovi{\'{c}} et~al\mbox{.}}{2021}]%
        {Dragan}
\bibfield{author}{\bibinfo{person}{{\v{S}}ejla Dautovi{\'{c}}},
  \bibinfo{person}{Dragan Doder}, {and} \bibinfo{person}{Zoran
  Ognjanovi{\'{c}}}.} \bibinfo{year}{2021}\natexlab{}.
\newblock \showarticletitle{An Epistemic Probabilistic Logic with Conditional
  Probabilities}. In \bibinfo{booktitle}{\emph{Proceedings of the 17th
  {JELIA}}} \emph{(\bibinfo{series}{LNCS}, Vol.~\bibinfo{volume}{12678})},
  \bibfield{editor}{\bibinfo{person}{Wolfgang Faber}, \bibinfo{person}{Gerhard
  Friedrich}, \bibinfo{person}{Martin Gebser}, {and} \bibinfo{person}{Michael
  Morak}} (Eds.). \bibinfo{publisher}{Springer}, \bibinfo{pages}{279--293}.
\newblock
\urldef\tempurl%
\url{https://doi.org/10.1007/978-3-030-75775-5\_19}
\showDOI{\tempurl}


\bibitem[\protect\citeauthoryear{Delgrande, Renne, and Sack}{Delgrande
  et~al\mbox{.}}{2019}]%
        {DELGRANDE_Rene_Sack2019}
\bibfield{author}{\bibinfo{person}{James~P. Delgrande}, \bibinfo{person}{Bryan
  Renne}, {and} \bibinfo{person}{Joshua Sack}.}
  \bibinfo{year}{2019}\natexlab{}.
\newblock \showarticletitle{The logic of qualitative probability}.
\newblock \bibinfo{journal}{\emph{Artificial Intelligence}}
  \bibinfo{volume}{275} (\bibinfo{year}{2019}), \bibinfo{pages}{457--486}.
\newblock
\showISSN{0004-3702}
\urldef\tempurl%
\url{https://doi.org/10.1016/j.artint.2019.07.002}
\showDOI{\tempurl}


\bibitem[\protect\citeauthoryear{{Della Monica}, Napoli, and Parente}{{Della
  Monica} et~al\mbox{.}}{2011}]%
        {DELLAMONICA2011}
\bibfield{author}{\bibinfo{person}{Dario {Della Monica}},
  \bibinfo{person}{Margherita Napoli}, {and} \bibinfo{person}{Mimmo Parente}.}
  \bibinfo{year}{2011}\natexlab{}.
\newblock \showarticletitle{On a Logic for Coalitional Games with
  Priced-Resource Agents}. In \bibinfo{booktitle}{\emph{Proceedings of the 7th
  {M4M} and the 4th {LAMAS}}} \emph{(\bibinfo{series}{ENTCS},
  Vol.~\bibinfo{volume}{278})}, \bibfield{editor}{\bibinfo{person}{Hans van
  Ditmarsch}, \bibinfo{person}{David Fern{\'{a}}ndez{-}Duque},
  \bibinfo{person}{Valentin Goranko}, \bibinfo{person}{Wojciech Jamroga}, {and}
  \bibinfo{person}{Manuel Ojeda{-}Aciego}} (Eds.).
  \bibinfo{publisher}{Elsevier}, \bibinfo{pages}{215--228}.
\newblock
\urldef\tempurl%
\url{https://doi.org/10.1016/j.entcs.2011.10.017}
\showDOI{\tempurl}


\bibitem[\protect\citeauthoryear{Dolgorukov and Gladyshev}{Dolgorukov and
  Gladyshev}{2023}]%
        {Dali2022}
\bibfield{author}{\bibinfo{person}{Vitaliy Dolgorukov} {and}
  \bibinfo{person}{Maksim Gladyshev}.} \bibinfo{year}{2023}\natexlab{}.
\newblock \showarticletitle{Dynamic Epistemic Logic for Budget-Constrained
  Agents}. In \bibinfo{booktitle}{\emph{Proceedings of the 4th {DaL{\'i}}}}
  \emph{(\bibinfo{series}{LNCS}, Vol.~\bibinfo{volume}{13780})},
  \bibfield{editor}{\bibinfo{person}{Carlos Areces} {and}
  \bibinfo{person}{Diana Costa}} (Eds.). \bibinfo{publisher}{Springer},
  \bibinfo{pages}{56--72}.
\newblock


\bibitem[\protect\citeauthoryear{Duc}{Duc}{1997}]%
        {Duc1997}
\bibfield{author}{\bibinfo{person}{Ho~Ngoc Duc}.}
  \bibinfo{year}{1997}\natexlab{}.
\newblock \showarticletitle{{Reasoning about Rational, but not Logically
  Omniscient, Agents}}.
\newblock \bibinfo{journal}{\emph{Journal of Logic and Computation}}
  \bibinfo{volume}{7}, \bibinfo{number}{5} (\bibinfo{year}{1997}),
  \bibinfo{pages}{633--648}.
\newblock
\urldef\tempurl%
\url{https://doi.org/10.1093/logcom/7.5.633}
\showDOI{\tempurl}


\bibitem[\protect\citeauthoryear{Engesser, Bolander, Mattm{\"{u}}ller, and
  Nebel}{Engesser et~al\mbox{.}}{2017}]%
        {Engesser2017}
\bibfield{author}{\bibinfo{person}{Thorsten Engesser}, \bibinfo{person}{Thomas
  Bolander}, \bibinfo{person}{Robert Mattm{\"{u}}ller}, {and}
  \bibinfo{person}{Bernhard Nebel}.} \bibinfo{year}{2017}\natexlab{}.
\newblock \showarticletitle{Cooperative Epistemic Multi-Agent Planning for
  Implicit Coordination}. In \bibinfo{booktitle}{\emph{Proceedings of the 9th
  {M4M@ICLA}}} \emph{(\bibinfo{series}{{EPTCS}}, Vol.~\bibinfo{volume}{243})},
  \bibfield{editor}{\bibinfo{person}{Sujata Ghosh} {and}
  \bibinfo{person}{R.~Ramanujam}} (Eds.). \bibinfo{pages}{75--90}.
\newblock
\urldef\tempurl%
\url{https://doi.org/10.4204/EPTCS.243.6}
\showDOI{\tempurl}


\bibitem[\protect\citeauthoryear{Fagin and Halpern}{Fagin and Halpern}{1987}]%
        {FAGIN1987}
\bibfield{author}{\bibinfo{person}{Ronald Fagin} {and}
  \bibinfo{person}{Joseph~Y. Halpern}.} \bibinfo{year}{1987}\natexlab{}.
\newblock \showarticletitle{Belief, awareness, and limited reasoning}.
\newblock \bibinfo{journal}{\emph{Artificial Intelligence}}
  \bibinfo{volume}{34}, \bibinfo{number}{1} (\bibinfo{year}{1987}),
  \bibinfo{pages}{39--76}.
\newblock
\showISSN{0004-3702}
\urldef\tempurl%
\url{https://doi.org/10.1016/0004-3702(87)90003-8}
\showDOI{\tempurl}


\bibitem[\protect\citeauthoryear{Fagin and Halpern}{Fagin and Halpern}{1994}]%
        {Fagin_Halpern1994}
\bibfield{author}{\bibinfo{person}{Ronald Fagin} {and}
  \bibinfo{person}{Joseph~Y. Halpern}.} \bibinfo{year}{1994}\natexlab{}.
\newblock \showarticletitle{Reasoning about Knowledge and Probability}.
\newblock \bibinfo{journal}{\emph{J. ACM}} \bibinfo{volume}{41},
  \bibinfo{number}{2} (\bibinfo{year}{1994}), \bibinfo{pages}{340–367}.
\newblock
\showISSN{0004-5411}
\urldef\tempurl%
\url{https://doi.org/10.1145/174652.174658}
\showDOI{\tempurl}


\bibitem[\protect\citeauthoryear{Fagin, Halpern, and Megiddo}{Fagin
  et~al\mbox{.}}{1990}]%
        {HalpernInequalities}
\bibfield{author}{\bibinfo{person}{Ronald Fagin}, \bibinfo{person}{Joseph~Y.
  Halpern}, {and} \bibinfo{person}{Nimrod Megiddo}.}
  \bibinfo{year}{1990}\natexlab{}.
\newblock \showarticletitle{A logic for reasoning about probabilities}.
\newblock \bibinfo{journal}{\emph{Information and Computation}}
  \bibinfo{volume}{87}, \bibinfo{number}{1} (\bibinfo{year}{1990}),
  \bibinfo{pages}{78--128}.
\newblock
\showISSN{0890-5401}
\urldef\tempurl%
\url{https://doi.org/10.1016/0890-5401(90)90060-U}
\showDOI{\tempurl}


\bibitem[\protect\citeauthoryear{Fagin, Halpern, Moses, and Vardi}{Fagin
  et~al\mbox{.}}{2003}]%
        {HalpernBook}
\bibfield{author}{\bibinfo{person}{Ronald Fagin}, \bibinfo{person}{Joseph~Y.
  Halpern}, \bibinfo{person}{Yoram Moses}, {and} \bibinfo{person}{Moshe
  Vardi}.} \bibinfo{year}{2003}\natexlab{}.
\newblock \bibinfo{booktitle}{\emph{Reasoning About Knowledge}}.
\newblock \bibinfo{publisher}{MIT Press}.
\newblock


\bibitem[\protect\citeauthoryear{Fan, Wang, and Ditmarsch}{Fan
  et~al\mbox{.}}{2015}]%
        {fan_wang_ditmarsch_2015}
\bibfield{author}{\bibinfo{person}{Jie Fan}, \bibinfo{person}{Yanjing Wang},
  {and} \bibinfo{person}{Hans~van Ditmarsch}.} \bibinfo{year}{2015}\natexlab{}.
\newblock \showarticletitle{Contingency and knowing whether}.
\newblock \bibinfo{journal}{\emph{The Review of Symbolic Logic}}
  \bibinfo{volume}{8}, \bibinfo{number}{1} (\bibinfo{year}{2015}),
  \bibinfo{pages}{75–107}.
\newblock
\urldef\tempurl%
\url{https://doi.org/10.1017/S1755020314000343}
\showDOI{\tempurl}


\bibitem[\protect\citeauthoryear{Feldman}{Feldman}{2000}]%
        {Feldman2000}
\bibfield{author}{\bibinfo{person}{Jacob Feldman}.}
  \bibinfo{year}{2000}\natexlab{}.
\newblock \showarticletitle{Minimization of Boolean complexity in human concept
  learning}.
\newblock \bibinfo{journal}{\emph{Nature}} \bibinfo{volume}{407},
  \bibinfo{number}{6804} (\bibinfo{year}{2000}), \bibinfo{pages}{630--633}.
\newblock
\showISSN{1476-4687}
\urldef\tempurl%
\url{https://doi.org/10.1038/35036586}
\showDOI{\tempurl}


\bibitem[\protect\citeauthoryear{Fischer and Ladner}{Fischer and
  Ladner}{1977}]%
        {Fischer_Ladner_1977}
\bibfield{author}{\bibinfo{person}{Michael~J. Fischer} {and}
  \bibinfo{person}{Richard~E. Ladner}.} \bibinfo{year}{1977}\natexlab{}.
\newblock \showarticletitle{Propositional Modal Logic of Programs}. In
  \bibinfo{booktitle}{\emph{Proceedings of the 9th {STOC}}}.
  \bibinfo{publisher}{{ACM}}, \bibinfo{pages}{286–294}.
\newblock
\showISBNx{9781450374095}
\urldef\tempurl%
\url{https://doi.org/10.1145/800105.803418}
\showDOI{\tempurl}


\bibitem[\protect\citeauthoryear{Galimullin}{Galimullin}{2021}]%
        {Galimullin2021}
\bibfield{author}{\bibinfo{person}{Rustam Galimullin}.}
  \bibinfo{year}{2021}\natexlab{}.
\newblock \showarticletitle{Coalition and Relativised Group Announcement
  Logic}.
\newblock \bibinfo{journal}{\emph{Journal of Logic, Language and Information}}
  \bibinfo{volume}{30}, \bibinfo{number}{3} (\bibinfo{year}{2021}),
  \bibinfo{pages}{451--489}.
\newblock
\urldef\tempurl%
\url{https://doi.org/10.1007/s10849-020-09327-2}
\showDOI{\tempurl}


\bibitem[\protect\citeauthoryear{Galimullin and
  Vel{\'a}zquez-Quesada}{Galimullin and Vel{\'a}zquez-Quesada}{2023}]%
        {Galimullin_2023}
\bibfield{author}{\bibinfo{person}{Rustam Galimullin} {and}
  \bibinfo{person}{Fernando~R. Vel{\'a}zquez-Quesada}.}
  \bibinfo{year}{2023}\natexlab{}.
\newblock \showarticletitle{({A}rbitrary) Partial Communication}. In
  \bibinfo{booktitle}{\emph{Proceedings of the 22nd {AAMAS}}},
  \bibfield{editor}{\bibinfo{person}{Noa Agmon}, \bibinfo{person}{Bo~An},
  \bibinfo{person}{Alessandro Ricci}, {and} \bibinfo{person}{William Yeoh}}
  (Eds.). \bibinfo{publisher}{{ACM}}, \bibinfo{pages}{400--408}.
\newblock
\urldef\tempurl%
\url{https://doi.org/10.5555/3545946.3598663}
\showDOI{\tempurl}


\bibitem[\protect\citeauthoryear{Halpern and Moses}{Halpern and Moses}{1992}]%
        {HALPERN1992}
\bibfield{author}{\bibinfo{person}{Joseph~Y. Halpern} {and}
  \bibinfo{person}{Yoram Moses}.} \bibinfo{year}{1992}\natexlab{}.
\newblock \showarticletitle{A guide to completeness and complexity for modal
  logics of knowledge and belief}.
\newblock \bibinfo{journal}{\emph{Artificial Intelligence}}
  \bibinfo{volume}{54}, \bibinfo{number}{3} (\bibinfo{year}{1992}),
  \bibinfo{pages}{319--379}.
\newblock
\showISSN{0004-3702}
\urldef\tempurl%
\url{https://doi.org/10.1016/0004-3702(92)90049-4}
\showDOI{\tempurl}


\bibitem[\protect\citeauthoryear{Khachiyan}{Khachiyan}{1980}]%
        {KHACHIYAN1980}
\bibfield{author}{\bibinfo{person}{L.G. Khachiyan}.}
  \bibinfo{year}{1980}\natexlab{}.
\newblock \showarticletitle{Polynomial algorithms in linear programming}.
\newblock \bibinfo{journal}{\emph{U. S. S. R. Comput. Math. and Math. Phys.}}
  \bibinfo{volume}{20}, \bibinfo{number}{1} (\bibinfo{year}{1980}),
  \bibinfo{pages}{53--72}.
\newblock
\showISSN{0041-5553}
\urldef\tempurl%
\url{https://doi.org/10.1016/0041-5553(80)90061-0}
\showDOI{\tempurl}


\bibitem[\protect\citeauthoryear{Naumov and Tao}{Naumov and Tao}{2015}]%
        {Naumov15}
\bibfield{author}{\bibinfo{person}{Pavel Naumov} {and} \bibinfo{person}{Jia
  Tao}.} \bibinfo{year}{2015}\natexlab{}.
\newblock \showarticletitle{Budget-Constrained Knowledge in Multiagent
  Systems}. In \bibinfo{booktitle}{\emph{Proceedings of the 14th {AAMAS}}},
  \bibfield{editor}{\bibinfo{person}{Gerhard Weiss}, \bibinfo{person}{Pinar
  Yolum}, \bibinfo{person}{Rafael~H. Bordini}, {and} \bibinfo{person}{Edith
  Elkind}} (Eds.). \bibinfo{publisher}{{IFAAMAS}}, \bibinfo{pages}{219–226}.
\newblock
\showISBNx{9781450334136}


\bibitem[\protect\citeauthoryear{Nguyen, Alechina, Logan, and Rakib}{Nguyen
  et~al\mbox{.}}{2015}]%
        {Nguyen_2015}
\bibfield{author}{\bibinfo{person}{Hoang~Nga Nguyen}, \bibinfo{person}{Natasha
  Alechina}, \bibinfo{person}{Brian Logan}, {and} \bibinfo{person}{Abdur
  Rakib}.} \bibinfo{year}{2015}\natexlab{}.
\newblock \showarticletitle{{Alternating-time temporal logic with resource
  bounds}}.
\newblock \bibinfo{journal}{\emph{Journal of Logic and Computation}}
  \bibinfo{volume}{28}, \bibinfo{number}{4} (\bibinfo{year}{2015}),
  \bibinfo{pages}{631--663}.
\newblock
\showISSN{0955-792X}
\urldef\tempurl%
\url{https://doi.org/10.1093/logcom/exv034}
\showDOI{\tempurl}


\bibitem[\protect\citeauthoryear{Pauly}{Pauly}{2002}]%
        {pauly02}
\bibfield{author}{\bibinfo{person}{Marc Pauly}.}
  \bibinfo{year}{2002}\natexlab{}.
\newblock \showarticletitle{A Modal Logic for Coalitional Power in Games}.
\newblock \bibinfo{journal}{\emph{Journal of Logic and Computation}}
  \bibinfo{volume}{12}, \bibinfo{number}{1} (\bibinfo{year}{2002}),
  \bibinfo{pages}{149--166}.
\newblock
\urldef\tempurl%
\url{https://doi.org/10.1093/logcom/12.1.149}
\showDOI{\tempurl}


\bibitem[\protect\citeauthoryear{Solaki}{Solaki}{2020}]%
        {Solaki2020}
\bibfield{author}{\bibinfo{person}{Anthia Solaki}.}
  \bibinfo{year}{2020}\natexlab{}.
\newblock \showarticletitle{Bounded Multi-agent Reasoning: Actualizing
  Distributed Knowledge}. In \bibinfo{booktitle}{\emph{Proceedings of the 3rd
  DaL{\'{\i}}}}, \bibfield{editor}{\bibinfo{person}{Manuel~A. Martins} {and}
  \bibinfo{person}{Igor Sedl{\'a}r}} (Eds.). \bibinfo{publisher}{Springer},
  \bibinfo{pages}{239--258}.
\newblock


\bibitem[\protect\citeauthoryear{Solaki}{Solaki}{2022}]%
        {Solaki2022}
\bibfield{author}{\bibinfo{person}{Anthia Solaki}.}
  \bibinfo{year}{2022}\natexlab{}.
\newblock \showarticletitle{{Actualizing distributed knowledge in bounded
  groups}}.
\newblock \bibinfo{journal}{\emph{Journal of Logic and Computation}}
  \bibinfo{volume}{33}, \bibinfo{number}{6} (\bibinfo{year}{2022}),
  \bibinfo{pages}{1497--1525}.
\newblock
\showISSN{0955-792X}
\urldef\tempurl%
\url{https://doi.org/10.1093/logcom/exac007}
\showDOI{\tempurl}


\bibitem[\protect\citeauthoryear{van Benthem}{van Benthem}{2007}]%
        {vanBenthem2007-VANDLF}
\bibfield{author}{\bibinfo{person}{Johan van Benthem}.}
  \bibinfo{year}{2007}\natexlab{}.
\newblock \showarticletitle{Dynamic Logic for Belief Revision}.
\newblock \bibinfo{journal}{\emph{Journal of Applied Non-Classical Logics}}
  \bibinfo{volume}{17}, \bibinfo{number}{2} (\bibinfo{year}{2007}),
  \bibinfo{pages}{129--155}.
\newblock
\urldef\tempurl%
\url{https://doi.org/10.3166/jancl.17.129-155}
\showDOI{\tempurl}


\bibitem[\protect\citeauthoryear{van Ditmarsch and Fan}{van Ditmarsch and
  Fan}{2016}]%
        {vanDitmarsch-Fan}
\bibfield{author}{\bibinfo{person}{Hans van Ditmarsch} {and}
  \bibinfo{person}{Jie Fan}.} \bibinfo{year}{2016}\natexlab{}.
\newblock \showarticletitle{Propositional quantification in logics of
  contingency}.
\newblock \bibinfo{journal}{\emph{Journal of Applied Non-Classical Logics}}
  \bibinfo{volume}{26}, \bibinfo{number}{1} (\bibinfo{year}{2016}),
  \bibinfo{pages}{81--102}.
\newblock
\urldef\tempurl%
\url{https://doi.org/10.1080/11663081.2016.1184931}
\showDOI{\tempurl}


\bibitem[\protect\citeauthoryear{van Ditmarsch, Hoek, and Kooi}{van Ditmarsch
  et~al\mbox{.}}{2008}]%
        {DELbook}
\bibfield{author}{\bibinfo{person}{Hans van Ditmarsch}, \bibinfo{person}{Wiebe
  Hoek}, {and} \bibinfo{person}{Barteld Kooi}.}
  \bibinfo{year}{2008}\natexlab{}.
\newblock \bibinfo{booktitle}{\emph{Dynamic Epistemic Logic}}.
\newblock \bibinfo{publisher}{Springer}.
\newblock
\urldef\tempurl%
\url{https://doi.org/10.1007/978-1-4020-5839-4}
\showDOI{\tempurl}


\bibitem[\protect\citeauthoryear{Wang and Cao}{Wang and Cao}{2013}]%
        {Wang2013}
\bibfield{author}{\bibinfo{person}{Yanjing Wang} {and}
  \bibinfo{person}{Qinxiang Cao}.} \bibinfo{year}{2013}\natexlab{}.
\newblock \showarticletitle{On Axiomatizations of Public Announcement Logic}.
\newblock \bibinfo{journal}{\emph{Synthese}} \bibinfo{volume}{190},
  \bibinfo{number}{Supplement-1} (\bibinfo{year}{2013}).
\newblock
\urldef\tempurl%
\url{https://doi.org/10.1007/s11229-012-0233-5}
\showDOI{\tempurl}


\end{thebibliography}

\clearpage
\appendix
\section*{A.1 Technical Appendix}
\textbf{Theorem 3.1 (Soundness).} The axiomatisation of $\ourlogic$ is sound.

\begin{proof} To prove Soundness of $\ourlogic$, we need to demonstrate that all axioms presented in \Cref{tab:ax} are valid and all inference rules preserve validity. The cases for (I1)--(I6) \cite{HalpernInequalities}, (K)--(C) and (Nec$_{i}$)--(RC1) \cite{HalpernBook,HALPERN1992} are standard and we refer to original papers. The cases for (B$^+$)--(c$^\approx$), (r$_p$) and (r$_\neg$)--(r$_{K2}$) are fairly simple and we omit them. Consider the remaining cases.

\textbf{Case (r$_{\ge}$).} 

\[\cM, w\vDash[?_G^A]\bigl(\sum\limits_{i=1}^{k}a_it_i\geq z\bigr) \text{ iff }\cM, w\vDash\text{BCS}(G, A)\to \bigl(\sum\limits_{i=1}^{k}a_it_i\geq z\bigr)^{(G, A)}\]

Note that $\cM, w\vDash [?_G^A] (z_1t_1+\dots + z_nt_n) \ge z$ iff $\cM, w\vDash \text{BCS}(G, A)$ implies $\cM^{?_G^A}, w\vDash (z_1t_1+\dots + z_nt_n) \ge z$ by \Cref{def:semantics}. Note also that $\cM^{?_G^A}, w\vDash (z_1t_1+\dots + z_nt_n) \ge z$ is equivalent to $\cM, w\vDash (z_1t^*_1+\dots + z_nt^*_n) \ge z$, where $t^*_k = t_k$ for $t_k = c_A$ (for all $A\in\cL_{PL}$) and $t_k = b_j$ (for all $j\notin G$) and if $t_k = b_i$ (for $i\in G$), then $t^*_k = t_k-\frac{\min_{j\in G}(c_j(A))}{|G|}$. The previous claim holds since  $\Cost_i^{?_G^A}(B) = \Cost_i(B)$, $\Bdg^{?_G^A}_j(w) = \Bdg_j(w)$ for $i\neq j$ and $\Bdg^{?_G^A}_i(w) = \Bdg_i(w) - \frac{\min_{j\in G}(\Cost_j(A))}{|G|}$. Then $\cM, w\vDash [?_G^A] (z_1t_1+\dots + z_nt_n) \ge z$ iff $\cM, w\vDash \text{BCS}(G, A)$ implies $\cM, w\vDash[(z_1t_1+\dots + z_nt_n) \ge z)]^{(G, A)}$. 

\textbf{Case (r$_{K1}$).}  
\[\cM, w\vDash[?_G^A]K_j \varphi \text{ iff } \cM, w\vDash \text{BCS}(G, A) \to K_j [?_G^A]\varphi, \text{for } j\notin G \]

($\Rightarrow$) Let $\cM, w\vDash [?_G^A]K_j \varphi$ (1) and $\cM, w\vDash \text{BCS}(G, A)$ (2). From (1), $\cM, w\vDash \text{BCS}(G, A)$ implies $\cM^{?_G^A}, w\vDash K_j \varphi$ by \Cref{def:semantics}. Then $\cM^{?_G^A}, w\vDash K_j \varphi$ from (2).  Then $\forall w': w\sim^{?_G^A}_j w' \Rightarrow \cM^{?_G^A}, w'\vDash \varphi$ by \Cref{def:semantics}. This, together with the fact that $w\sim^{?_G^A}_jw'$ iff $w\sim_jw'$, $\cM, w\vDash \text{BCS}(G, A)$ and $\cM, w'\vDash \text{BCS}(G, A)$,  implies that
$\forall w',$ s.t. $w\sim_j w'$: $\cM, w'\vDash \text{BCS}(G, A) \Rightarrow \cM^{?_G^A}, w'\vDash \varphi$. This is equivalent to $\cM, w\vDash K_j[?_iA]\varphi$, by \Cref{def:semantics}.

($\Leftarrow$) The case for $\cM, w\nvDash \text{BCS}(G, A)$ is trivial. Consider only the case for $\cM, w\vDash \text{BCS}(G, A) \wedge K_j[?_iA]\varphi$. Then $\forall w'$ s.t. $ w \sim_j w', \cM, w'\vDash \text{BCS}(G, A) \Rightarrow \cM^{?_G^A}, w'\vDash \varphi$. Then $\forall w', w\sim^{?_G^A}_jw' \Rightarrow \cM^{?_G^A}, w'\vDash \varphi$. By \Cref{def:semantics}, $\cM^{?_G^A}, w\vDash K_j\varphi$ and hence $\cM, w\vDash [?_G^A]K_j\varphi$.

\textbf{Case (r$_{K2}$).} 
\[\cM, w\vDash[?_G^A]K_i \varphi \text{ iff } \cM, w\vDash\text{BCS}(G, A) \to\]\[\to\bigwedge\limits_{A'\in \{A, \neg A\}}\Big(\bigl(A' \to K_i(A' \to [?_G^A]\varphi)\bigr)\Big), \text{where } i\in G\] 

($\Rightarrow$)  Let $\cM, w\vDash [?_G^A]K_i\varphi$ (1) and $\cM, w\vDash \text{BCS}(G, A)$ (2).
From (1), (2) and \Cref{def:semantics} we get $\cM^{?_G^A}, w\vDash K_i \varphi$.
Then $\forall w', w\sim^{?_G^A}_iw' \Rightarrow \cM^{?_G^A}, {w'\vDash \varphi}$. Assume that $\cM, w\vDash A$. Then $\forall w'$, s.t. $w\sim_i w' \text{ and } \cM, w'\vDash A \text{ and } \cM, w'\vDash \text{BCS}(G, A)$ implies $\cM^{?_G^A}, w'\vDash \varphi$. This is equivalent to $\cM, w\vDash K_i (A\to [?_G^A]\varphi)$ by \Cref{def:semantics}. Then, from our assumption we have proved that $\cM, w\vDash A \to K_i (A\to [?_G^A]\varphi)$. By a symmetric argument, we can show that $\cM, w\vDash \neg A \to K_i (\neg A\to [?_G^A\varphi)$.

($\Leftarrow$) The case for $\cM, w\nvDash \text{BCS}(G, A)$ is trivial. Consider only the case for $\cM, w\vDash \text{BCS}(G, A) \wedge \bigwedge\limits_{A' \in \{A, \neg A \}}
\bigl(A' \to K_i(A' \to [?_G^A]\varphi)\bigr)$. Assume that $\cM, w\vDash A$. Then $\cM, w\vDash K_i(A\to [?_G^A]\varphi)$.  Similarly, assuming $\cM, w\vDash \neg A$ entails $\cM, w\vDash K_i(\neg A\to [?_G^A]\varphi)$. Then for all $w'$, s.t. $w\sim_i w'$ and $w'$ agrees with $w$ on the valuation of $A$ it holds that $\cM, w'\vDash [?_G^A]\varphi$ and hence  $\cM, w'\vDash \text{BCS}(G, A)$ implies $\cM^{?_G^A}, w'\vDash \varphi$. Then it holds that $\forall w': w\sim^{?_G^A}_iw' \Rightarrow \cM^{?_G^A}, w'\vDash \varphi$. And hence $\cM^{?_G^A}, w\vDash K_i\varphi$. By \Cref{def:semantics}, the last claim implies $\cM, w\vDash [?_G^A]K_i\varphi$.

\textbf{Case (RC2).} \[\cM\vDash\chi\to[?_G^A]\psi \text{ and }\]
\[\cM\vDash (\chi \wedge \text{BCS}(G, A))\to\bigwedge\limits_{A'\in \{A, \neg A\}}(A'\to E_{H\cap G}(A'\to\chi))\wedge E_{H\setminus G}\chi,\]
\[\text{ implies } \cM\vDash \chi\to [?_G^A]C_H\psi\]

Since cases for $H\cap G = \emptyset$ and $H\setminus G = \emptyset$ are trivial, consider only the case for $H\cap G \neq \emptyset$ and $H\setminus G \neq \emptyset$. From the antecedent we know that $\cM, w\vDash \chi$ implies $\cM^{?_G^A}, w\vDash \psi$ and  $\cM^{?_G^A}, w\vDash E_H\chi$ for all $w\in W^{?_G^A}$. The latter implies $\cM^{?_G^A}, w\vDash \chi$ by the reflexivity of $\sim$. Since $\cM^{?_G^A}, w\vDash \psi$ for all $w\in W^{?_G^A}$, it also holds that $\cM^{?_G^A}, w'\vDash \psi$ for all $w'\sim^{?_G^A}_H w$ and then $\cM^{?_G^A}, w'\vDash E_H\psi$. That gives us $\cM^{?_G^A}, w\vDash \chi \to E_H(\chi \wedge \psi)$. By RC1 we get $\cM^{?_G^A}, w\vDash \chi \to C_H\psi$ and then $\cM, w\vDash \chi$ implies $\cM^{?_G^A}, w\vDash \chi \to C_H\psi$ for all $w\in W^{?_G^A}$. Together with the fact that $\cM^{?_G^A}, w\vDash \chi$ it gives us $\cM, w\vDash \chi \to [?_A^G]C_H\psi$.
\end{proof}

\textbf{Lemma 3.6 (Truth Lemma).} Let $\cM^c$ be the canonical model for $\varphi$. For all $\psi\in cl(\varphi), w \in W^c:  \cM^c, w \vDash \psi$ iff $\psi\in w$. 

\begin{proof}

Note that the proof of the Truth Lemma is organized by the induction on a \emph{complexity} measure for $\varphi$, denoted $c(\varphi)$. We define this \emph{complexity} measures as follows.

\textbf{Definition ( Formula Complexity).} The complexity $c: \cL_{\ourlogic} \to  \mathbb{N} $ of $\ourlogic$ formulas is defined as follows: \\
(1) $c(p) = 1$;\\
(2) $c((z_1t_1+ \dots + z_nt_n)\geq z) = 1$;\\
(3) $c(\neg \varphi) = c(\varphi) + 1$;\\
(4) $c(\varphi \wedge \psi) = \max\{ c(\varphi), c(\psi) \} + 1$;\\
(5) $c(K_i \varphi) = c(\varphi) + 1$;\\
(6) $c(C_G \varphi) = c(\varphi) + 1$;\\
(7) $c([?_G^A]\varphi) = (c(A)+ 5) \cdot c(\varphi)$; 

The number 5 in this Definition seems arbitrary, but we may take take any natural number that gives us the
following properties.

\textbf{Proposition.} For all $\psi, \varphi, \chi \in \cL(\ourlogic)$ and all $A\in \cL_{PL}$ 
\begin{enumerate}
\item $c(\psi) \geq c(\varphi)$, for all $\varphi\in\Sub(\psi)$ 
\item $c([?_G^A]p) > c(\text{BCS}(G, A)\to p)$ 
\item  $c([_GA]\bigl(\sum\limits_{i=1}^{k}a_it_i\geq z\bigr)) >  c\Big(\Big(\text{BCS}(G, A) \to \bigl(\sum\limits_{i=1}^{k}a_it_i\geq z\bigr)^{(G, A)}\Big)\Big)$
\item $c([?_G^A]\neg \varphi) > c(\text{BCS}(G, A) \to \neg [?_G^A]\varphi)$ 
\item $c([?_G^A](\varphi \wedge \psi)) > c([?_G^A]\varphi \wedge [?_G^A]\psi)$
\item $c([?_G^A]K_j \varphi) > c(\text{BCS}(G, A) \to K_j [?_G^A]\varphi)$,
 where $j\notin G$ 
 \item $c([?_G^A]K_i \varphi \leftrightarrow \text{BCS}(G, A)) >
c(\bigwedge\limits_{A'\in \{A, \neg A\}}\Big(\bigl(A' \to K_i(A' \to [?_G^A]\varphi)\bigr)\Big))$,  where $i\in G$
\item $c([?_G^A]C_H\varphi)>c([?_G^A]\varphi)$
\end{enumerate}

 \end{proof}


\end{document}